\documentclass[]{IEEEtran}

\usepackage{graphicx}
 \graphicspath{{figures/}}
 \DeclareGraphicsExtensions{.pdf,.jpeg,.png,.jpg,}
 \usepackage{epstopdf}
 \usepackage[caption=false,font=footnotesize]{subfig}

\usepackage{amsmath,mathtools}
\usepackage{amssymb}
\usepackage{bm}

\usepackage{cite}

\usepackage{url}

\renewcommand{\Pr}{\mathbb{P}} 

\newcommand{\pmf}{\phi}
\newcommand{\cond}{\;\middle\vert\:}
\newcommand{\bvec}[1]{\bm{#1}}

\newcommand{\vdd}{V_\mathrm{dd}}
\newcommand{\pres}{p_\mathrm{res}}
\newcommand{\gammaseq}{\bvec{\overline{\gamma}}}

\newcommand{\volt}{\,\mathrm{V}}
\newcommand{\ns}{\,\mathrm{ns}}
\newcommand{\pJ}{\,\mathrm{pJ}}
\newcommand{\nJns}{\,\mathrm{nJ}\cdot\mathrm{ns}}

\usepackage{amsthm}
\newtheorem{lemma}{Lemma} 
\newtheorem{definition}{Definition}

\usepackage[noend,ruled]{algorithm2e}
	   \SetAlFnt{\small}

\usepackage{booktabs}

\newcommand\textvtt[1]{{\normalfont\fontfamily{cmvtt}\selectfont #1}}

\usepackage[acronym]{glossaries}
\newacronym{iid}{i.i.d.}{independent and identically distributed} 
\newacronym{MC}{MC}{Monte-Carlo}
\newacronym{exit}{ExIT}{``extrinsic information transfer''}
\newacronym{1D}{\mbox{1-D}}{one-dimensional} 
\newacronym{pmf}{PMF}{probability mass function}

\title{Modeling and Energy Optimization of\\ LDPC Decoder Circuits with Timing Violations}
\author{Fran\c{c}ois Leduc-Primeau, Frank~R.~Kschischang, and Warren~J.~Gross
	\thanks{F.~Leduc-Primeau is with the Department of Electronics, IMT Atlantique, Brest, France, and with CNRS Lab-STICC (e-mail: francois.leduc-primeau@imt-atlantique.fr). W.~J.~Gross is with the Department of Electrical and Computer Engineering, McGill University, Montreal, Canada (e-mail: warren.gross@mcgill.ca).
F.~R.~Kschischang is with the Department of Electrical and Computer Engineering, University of Toronto, Toronto, Canada (e-mail: frank@comm.utoronto.ca).
	}  
	\thanks{A preliminary version of this work was presented at the IEEE ICC 2015 conference~\cite{leduc-primeau:2015b}.}
}

\begin{document}
\maketitle
\thispagestyle{empty}

\begin{abstract}
This paper proposes a ``quasi-synchronous'' design approach for signal processing circuits, in which timing violations are permitted, but without the need for a hardware compensation mechanism.
The case of a low-density parity-check (LDPC) decoder is studied, and a method for accurately modeling the effect of timing violations at a high level of abstraction is presented.
The error-correction performance of code ensembles is then evaluated using density evolution while taking into account the effect of timing faults.
Following this, several quasi-synchronous LDPC decoder circuits based on the offset min-sum algorithm are optimized, providing a 23\%--40\% reduction in energy consumption or energy-delay product, while achieving the same performance and occupying the same area as conventional synchronous circuits.
\end{abstract}

\section{Introduction}\label{sec:intro}
The time required for a signal to propagate through a CMOS circuit varies depending on several factors. Some of the variation results from physical limitations: the delay depends on the initial and final charge state of the circuit. Other variations are due to the difficulty (or impossibility) of controlling the fabrication process and the operating conditions of the circuit \cite{ghosh:2010}.
As process technologies approach atomic scales, the magnitude of these variations is increasing, and reducing the supply voltage to save energy increases the variations even further \cite{dreslinski:2010}.

The variation in propagation delay is a source of energy inefficiency for synchronous circuits since the clock period is determined by the worst delay. One approach to alleviate this problem is to allow timing violations to occur.
While this would normally be catastrophic, some applications (in signal processing or in error-correcting decoding, for example) can tolerate a degradation in the operation of the circuit, either because an approximation to the ideal output suffices, or because the algorithm intrinsically rejects noise.  This paper proposes an approach to the design of
systems 
that are tolerant to timing violations.  In particular we apply this approach to the design of
energy-optimized low-density parity-check (LDPC) decoder circuits based on a
state-of-the-art soft-input algorithm and architecture.

Other approaches have been previously proposed to build synchronous systems that can tolerate some timing violations. 
In \emph{better than worst-case} (BTWC) \cite{austin:2005} or \emph{voltage over-scaled} (VOS) circuits, a mechanism is added to the circuit to compensate or recover from timing faults.
One such method introduces special latches that can detect timing violations, and can trigger a restart of the computation when needed \cite{bowman:2009,das:2009}. Since the circuit's latency is increased significantly when a timing violation occurs, this approach is only suitable for tolerating small fault rates (e.g., $10^{-7}$) and for applications where the circuit can be easily restarted, such as microprocessors that support speculative execution.

In most signal processing tasks, it is acceptable for the output to be non-deterministic, which creates more possibilities for dealing with timing violations.
A seminal contribution in this area was the algorithmic noise tolerance (ANT) approach \cite{hegde:1999,shim:2004}, which is to allow timing violations to occur in the main processing block, while adding a separate reliable processing block with reduced precision that is used to bound the error of the main block, and provide algorithmic performance guarantees.
The downside of the ANT approach is that it relies on the assumption that timing violations will first occur in the most significant bits. If that is not the case, the precision of the circuit can degrade to the precision of the auxiliary block, limiting the scheme's usefulness. For many circuits, including some adder circuits \cite{liu:2010}, this assumption does not hold. Furthermore, the addition of the reduced precision block and of a comparison circuit increases the area requirement.

We propose a design methodology for digital circuits with a relaxed synchronicity requirement that does not rely on any hardware compensation mechanism. Instead, we provide performance guarantees by re-analyzing the algorithm while taking into account the effect of timing violations. We say that such systems are \emph{quasi-synchronous}.
%
LDPC decoding algorithms are good candidates for a quasi-synchronous implementation because their throughput and energy consumption are limiting factors in many applications, and like other signal processing algorithms, 
their performance is assessed in terms of expected values.
Furthermore, since the algorithm is iterative, there is a possibility to optimize each iteration separately, and we show that this allows for additional energy savings.

The topic of unreliable LDPC decoders has been discussed in a number of contributions. Varshney studied the Gallager-A and the Sum-Product decoding algorithms when the computations and the message exchanges are ``noisy'', and showed that the density evolution analysis still applies \cite{varshney:2011}.
The Gallager-B algorithm was also analyzed under various scenarios \cite{leduc-primeau:2012,tabatabaei-yazdi:2013,huang:2014}.
A model for an unreliable quantized Min-Sum decoder was proposed in \cite{ngassa:2013}, which provided numerical evaluation of the density evolution equations as well as simulations of a finite-length decoder.
Faulty finite-alphabet decoders were studied in \cite{dupraz:2015}, where it was proposed to model the decoder messages using conditional distributions that depend on the ideal messages.
The quantized Min-Sum decoder was also analyzed in \cite{balatsoukas-stimming:2014} for the case where faults are the result of storing decoder messages in an unreliable memory.
The specific case of faults caused by delay variations in synchronous circuits is considered in \cite{brkic:2015}, where a deviation model is proposed for binary-output circuits in which a deviation occurs probabilistically when the output of a circuit changes from one clock cycle to the next, but cannot occur if the output does not change.
While none of these contributions explicitly consider the relationship between the reliability of the decoder's implementation and the energy it consumes, there have been some recent developments in the analysis of the energy consumption of reliable decoders. Lower bounds for the scaling of the energy consumption of error-correction decoders in terms of the code length are derived in \cite{blake:2015}, and tighter lower bounds that apply to LDPC decoders are derived in \cite{blake:2015b}.
The power required by regular LDPC decoders is also examined in \cite{ganesan:2016}, as part of the study of the total power required for transmitting and decoding the codewords.

In this paper, we present a modeling approach that provides an accurate representation of the deviations introduced in the output of an LDPC decoder processing circuit in the presence of occasional timing violations, while simultaneously measuring its energy consumption. 
We introduce a \emph{weak symmetry} property for this model, and show that when it is satisfied,
the model can be used as part of a density evolution analysis to evaluate the channel threshold and iterative performance of the decoder affected by timing faults.
We also present experimental evidence showing that weak symmetry is satisfied for the decoder circuits that we consider.
Finally, we show that under mild assumptions, the problem of minimizing the energy consumption of a quasi-synchronous decoder can be simplified to the energy minimization of a small test circuit, and present an approximate optimization method similar to Gear-Shift Decoding \cite{ardakani:2006} that finds sequences of quasi-synchronous decoders that minimize decoding energy subject to performance constraints.
We note that subsequent to \cite{leduc-primeau:2015b}, an energy optimization method for faulty LDPC decoders was presented in \cite{schlafer:2016}. Both methods look for a sequence of operating conditions that minimize decoding energy, but the method in \cite{schlafer:2016} requires that the operating conditions be ordered based on the message error probability that can be achieved, which is not possible in general without knowing the message distribution.

The remainder of the paper is organized as follows. 
Section~\ref{sec:LDPC} reviews LDPC codes and describes the circuit architecture of the decoder that is used to measure timing faults.
Section~\ref{sec:deviation} presents the \emph{deviation} model that represents the effect of timing faults on the algorithm.
Section~\ref{sec:analysis} then discusses the use of density evolution and of the deviation model to predict the performance of a decoder affected by timing faults.
Finally, Section~\ref{sec:optimization} presents the energy optimization strategy and results, and Section~\ref{sec:conclusion} concludes the paper.
Additional details on the CAD framework used for circuit measurements can be found in Appendix~\ref{sec:appendix:workflow}, and Appendix~\ref{sec:appendix:testcircuit} provides some details concerning the simulation of the test circuits.

\section{LDPC Decoding Algorithm and Architecture}\label{sec:LDPC}

	\subsection{Code and Channel}\label{sec:code-channel}

	We consider a communication scenario where a sequence of information bits is encoded using a binary LDPC code of length $n$. 
	The LDPC code described by an $m \times n$ binary parity-check matrix $H = [ h_{j,i} ]$ consists of all length-$n$ row vectors $v$ satisfying the equation $vH^\textsc{T} = 0$. Equivalently, the code can be described by a bipartite Tanner graph with $n$ \emph{variable nodes} (VN) and $m$ \emph{check nodes} (CN) having an edge between the $i$-th variable node and the $j$-th check node if and only if $h_{j,i} \neq 0$.
	We assume that the LDPC code is regular, which means that in the code's Tanner graph each variable node has a fixed degree $d_v$ and each check node has a fixed degree $d_c$.

	Let us assume that the transmission takes place over the binary-input additive white Gaussian noise (BIAWGN) channel. A codeword $\bvec{x} \in \{-1,1\}^n$ is transmitted through the channel, which outputs the received vector $\bvec{y} = \bvec{x} + \bvec{w}$, where $\bvec{w}$ is a vector of $n$ \gls{iid} zero-mean normal random variables with variance $\sigma^2$.
	We use $x_i$ and $y_i$ to refer to the input and output of the channel at time $i$.
	The BIAWGN channel has the property of being output symmetric, meaning that
	$\phi_{y_i \vert x_i}\left( q \cond 1 \right) = \phi_{y_i \vert x_i}\left( -q \cond -1 \right)$,
	and memoryless, meaning that
	$\phi_{\bm{y} \vert \bm{x}}\left(\bm{q} \cond \bm{r}\right) = \prod_{i=1}^n \phi_{y_i \vert x_i}\left(q_i \cond r_i\right)$.
	Throughout the paper, $\pmf(\cdot)$ denotes a probability density function.
	The BIAWGN channel can also be described multiplicatively as $\bm{y}=\bm{xz}$, where $\bm{z}$ is a vector of \gls{iid} normal random variables with mean $1$ and variance $\sigma^2$.
	
	Let the \emph{belief} output $\mu_i$ of the channel at time $i$ be given by 
	\begin{equation}\label{eq:channelbelief}
	\mu_i = \frac{\alpha y_i}{ \sigma^2} \, , 
	\end{equation}
	with $\alpha>0$.
	Note that if $\alpha=2$ then $\mu_i$ is a log-likelihood ratio.
	Assuming that $x_i=1$ was transmitted, then $\mu_i$ has a normal distribution with mean $\alpha/\sigma^2$ and variance $\alpha^2/\sigma^2$.
	Writing $\rho = \alpha/\sigma^2$, we see that $\mu_i$ is Gaussian with mean $\rho$ and variance $\alpha \rho$, that is, the distribution of $\mu_i$ is described by a single parameter $\rho$.
	We call this distribution a \gls{1D} normal distribution.
	The distribution of $\mu_i$ can also be specified using other equivalent parameters, such as the probability of error $p_e$, given by 
\ifCLASSOPTIONdraftcls 
	\begin{equation}\label{eq:errprob}
	p_e = \Pr\left(\mu_i < 0 \middle\vert x_i=1\right) 
	   = \Pr\left(\mu_i>0 \middle\vert x_i=-1\right)
	   = \frac{1}{2} \,\mathrm{erfc}\!\left(\frac{1}{\sqrt{2\sigma^2}} \right)= \frac{1}{2} \, \mathrm{erfc}\!\left(\sqrt{\frac{\rho}{2\alpha}}\right),
	\end{equation}
\else 
	\begin{multline}\label{eq:errprob}
	p_e = \Pr\left(\mu_i < 0 \middle\vert x_i=1\right) 
	   = \Pr\left(\mu_i>0 \middle\vert x_i=-1\right)\\
	   = \frac{1}{2} \,\mathrm{erfc}\!\left(\frac{1}{\sqrt{2\sigma^2}} \right)= \frac{1}{2} \, \mathrm{erfc}\!\left(\sqrt{\frac{\rho}{2\alpha}}\right),
	\end{multline}
\fi
	where $\mathrm{erfc}(\cdot)$ is the complementary error function.

	\subsection{Decoding Algorithm}\label{sec:algorithm}
	The well-known Offset Min-Sum (OMS) algorithm is a simplified version of the Sum-Product algorithm that can usually achieve similar error-correction performance. It has been widely used in implementations of LDPC decoders \cite{cushon:2010,roth:2010,sun:2011}.
	To make our decoder implementation more realistic and show the flexibility of our design framework, we present an algorithm and architecture that support a row-layered message-passing schedule. 
	Architectures optimized for this schedule have proven effective for achieving efficient implementations of LDPC decoders \cite{roth:2010,sun:2011,cevrero:2010}.
	Using a row-layered schedule also allows the decoder to be pipelined to increase the circuit's utilization.
	In a row-layered LDPC decoder, the rows of the parity-check matrix are partitioned into $L$ sets called \emph{layers}. 
	To simplify the description of the decoding algorithm, we assume that all the columns in a given layer contain exactly one non-zero element. This implies that $L=d_v$.
	Note that codes with \emph{at most} one non-zero element per column and per layer can also be supported by the same architecture, simply requiring a modification of the way algorithm variables are indexed.

	Let us define a set $\mathcal{L}_\ell$ containing the indices of the rows of $H$ that are part of layer $\ell$, $\ell \in [1,L]$. 
	We denote by $\mu_{i,j}^{(t)}$ a message sent from VN $i$ to CN $j$ during iteration $t$. and by $\lambda_{i,j}^{(t)}$ a message sent from CN $j$ to VN $i$. 
	It is also useful to refer to the CN neighbor of a VN $i$ that is part of layer $\ell$. Because of the restriction mentioned above, there is exactly one such CN, and we denote its index by $J(i,\ell)$.
	Finally, we denote the channel information corresponding to the $i$-th codeword bit by $\mu^{(0)}_i$, since it also corresponds to the first message sent by a variable node $i$ to all its neighboring check nodes.
	
	The Offset Min-Sum algorithm used with a row-layered message-passing schedule is described in Algorithm~\ref{alg:loms}.
	In the algorithm, $\mathcal{N}(j)$ denotes the set of indices corresponding to VNs that are neighbors of a check node $j$,
	and $\Lambda_i$ represents the current sum of incoming messages at a VN $i$.
	The function $\min_{1,2}(S)$ returns the smallest and second smallest values in the set $S$, 
	$C \geq 0$ is the offset parameter, and
	\begin{equation*}
	\mathrm{sgn}(x)=
	\begin{cases}
	1 & \text{if $x \geq 0$,} \\
	-1 & \text{if $x < 0$.}
	\end{cases}
	\end{equation*}

	\begin{algorithm}[tb] 
	   \SetKwInOut{Input}{input} \SetKwInOut{Output}{output}
	   \DontPrintSemicolon
	   \LinesNumbered
	   
	   \Input{$\{\mu^{(0)}_1, \mu^{(0)}_2, \ldots, \mu^{(0)}_n\}$}
	   \Output{$\{\hat{x}_1, \hat{x}_2, \ldots, \hat{x}_n\}$}
	   \Begin{
	   
	   	\tcp{Initialization}
		$\Lambda_i \gets \mu^{(0)}_i, \; \forall i \in [1,n]$\;
		$\lambda_{i,j}^{(0)} \gets 0, \; \forall j \in [1,m], i \in \mathcal{N}(j)$\;
	   
	   	\tcp{Decoding}
	   	\For{$t \gets 1$ \KwTo $T$ } {
			\For{$\ell \gets 1$ \KwTo $L$} { 
				\tcp{VN to CN messages}
				$\mu_{i,J(i,\ell)}^{(t)} \gets \Lambda_i - \lambda_{i,J(i,\ell)}^{(t-1)}, \; \forall i$\;			
				\tcp{CN to VN messages}
				\For{$j \in \mathcal{L}_\ell$} {
					$[m_1, m_2] \gets \min_{1,2}(\{ |\mu_{i,j}^{(t)}| : i \in \mathcal{N}(j)\})$\;
					$m_1 \gets \max(0, m_1 - C)$\;
					$m_2 \gets \max(0, m_2 - C)$\;
					$s_T \gets \prod_{i \in \mathcal{N}(j)} \mathrm{sgn}(\mu_{i,j}^{(t)})$\;
					\For{$i \in \mathcal{N}(j)$} {
						$s_i \gets s_T \cdot \mathrm{sgn}(\mu_{i,j}^{(t)})$\;
						\lIf{ $|\mu_{i,j}^{(t)}| = m_1$ } {
							$\lambda_{i,j}^{(t)}  \gets s_i \cdot m_2$
						}
						\lElse {
							$\lambda_{i,j}^{(t)}  \gets s_i \cdot m_1$
						}
					}
				}
				\tcp{VN update}
				$\Lambda_i \gets \mu_{i,J(i,\ell)}^{(t)} + \lambda_{i,J(i,\ell)}^{(t)}, \; \forall i$\;
				\tcp{VN decision} 
				\For{$i \in \{1,2,\dots,n\}$} {
					\lIf{$\Lambda_i > 0$} {
						$\hat{x}_i \gets 1$
					} 
					\lElseIf{$\Lambda_i < 0$} {
						$\hat{x}_i \gets -1$
					} 
					\lElse {
						$\hat{x}_i \gets 1 \,\mathrm{or}\, -1 \text{ with equal probability}$
					}
				}
			}	
		}
	   }
	\caption{OMS with a row-layered schedule.}
	\label{alg:loms}
	\end{algorithm}
	
	\subsection{Architecture}\label{sec:architecture}
	
	The Tanner graph of the code can also be used to represent the computations that must be performed by the decoder. At each decoding iteration, one message is sent from variable to check nodes on every edge of the graph, and again from check to variable nodes.
	We call a variable node processor (VNP) a circuit block that is responsible for generating messages sent by a variable node, and similarly a check node processor (CNP) a circuit block generating messages sent by a check node.
	
	In a row-layered architecture in which the column weight of layer subsets is at most 1, there is at most one message to be sent and received for each variable node in a given layer. Therefore VNPs are responsible for sending and receiving one message per clock cycle. CNPs on the other hand receive and send $d_c$ messages per clock cycle.	
	At any given time, every VNP and CNP is mapped respectively to a VN and a CN in the Tanner graph.
	The routing of messages from VNPs to CNPs and back can be posed as two equivalent problems. One can fix the mapping of VNs to VNPs and of CNs to CNPs, and find a permutation of the message sequence that matches VNP outputs to CNP inputs, and another permutation that matches CNP outputs to VNP inputs.
	Alternatively, if VNPs process only one message at a time, one can fix the connections between VNPs and CNPs, and choose the assignment of VN to VNPs to achieve correct message routing. 
	We choose the later approach because it allows studying the computation circuit without being concerned by the routing of messages.
	
	The number of CNPs instantiated in the decoder can be adjusted based on throughput requirements from $1$ to $m/L$ (the number of rows in a layer). As the number of CNPs is varied, the number of VNPs will vary from $d_c$ to $n$.
	An architecture diagram showing one VNP and one CNP is shown in Fig.~\ref{fig:layered_arch}. In reality, a CNP is connected to $d_c-1$ additional VNPs, which are not shown. The memories storing the belief totals $\Lambda_i$ and the intrinsic beliefs $\lambda_{i,j}^{(t)}$ are also not shown.
	The part of the VNP responsible for sending a message to the CNP is called VNP \emph{front} and the part responsible for processing a message received from a CNP is called the VNP \emph{back}. The VNP front and back do not have to be simultaneously mapped to the same VN. This allows to easily vary the number of pipeline stages in the VNPs and CNPs. Fig.~\ref{fig:layered_arch} shows the circuit with two pipeline stages. 
	
	Messages exchanged in the decoder are fixed-point numbers. The position of the binary point does not have an impact on the algorithm, and therefore the messages sent by VNs in the first iteration can be defined as rounding the result of \eqref{eq:channelbelief} to the nearest integer, while choosing a suitable $\alpha$. 
	The number of bits in the quantization, the scaling factor $\alpha$, and the OMS offset parameter are chosen based on a density evolution analysis of the algorithm (described in Section~\ref{sec:analysis}).
	We quantize decoder messages to 6~bits, which yields a decoder with approximately the same channel threshold as a floating-point decoder under a standard fault-free implementation. 
	
	\begin{figure}[tbp]
	\begin{center}
	\includegraphics[width=3.0in]{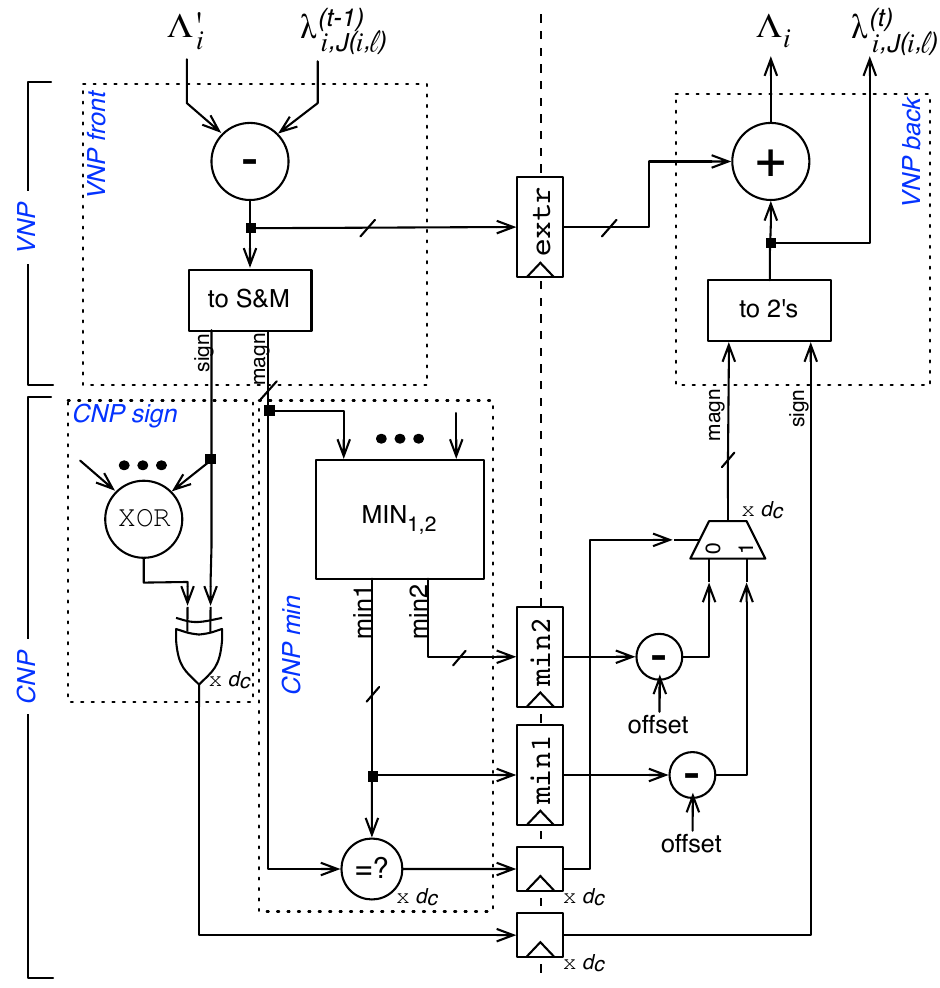}
	\caption{Block diagram of the layered Offset Min-Sum decoder architecture.}
	\label{fig:layered_arch}
	\end{center}
	\end{figure}
		
	In order to analyze a circuit that is representative of state-of-the-art architectures, we use an optimized architecture for finding the first two minima in each CNP. Our architecture is inspired by the ``tree structure'' approach presented in \cite{wey:2008}, but requires fewer comparators. Each pair of CNP inputs is first sorted using the \emph{Sort} block shown in Fig.~\ref{fig:sort-2}. These sorted pairs are then merged recursively using a tree of \emph{Merge} blocks, shown in Fig.~\ref{fig:merge}. If the number of CNP inputs is odd, the input that cannot be paired is fed directly into a special merge block with 3 inputs, which can be obtained from the 4-input \emph{Merge} block by removing the $\mathrm{min}_\mathrm{2b}$ input and the bottom multiplexer.
	
	Note that it is possible that changes to the architecture could increase or decrease the robustness of the decoder (see e.g.~\cite{sedighi:2014}), but this is outside the scope of this paper.
	
	\begin{figure}[tbp]
	\begin{center}
	\subfloat[][Sort block.]{\label{fig:sort-2}\includegraphics[scale=0.45]{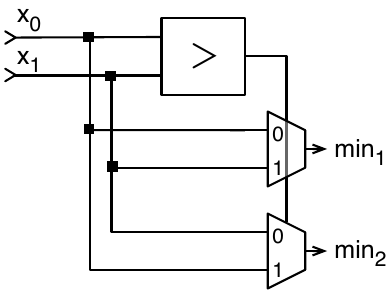}}
	\qquad
	\subfloat[][Merge block.]{\label{fig:merge}\includegraphics[scale=0.45]{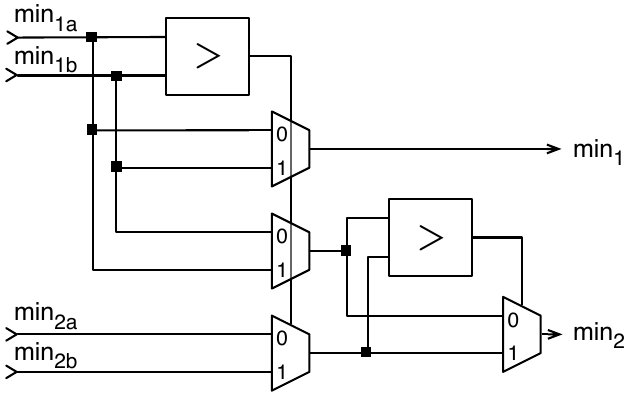}}
	\caption{Logic blocks used in the $\mathrm{MIN}_{1,2}$ unit.}
	\label{fig:minarch}
	\end{center}
	\end{figure}

\section{Deviation Model}\label{sec:deviation}

	
	\subsection{Quasi-Synchronous Systems}\label{sec:deviation:QS}
	We consider a synchronous system that permits timing violations without hardware compensation, resulting in what we call a \emph{quasi-synchronous} system. Optimizing the energy consumption of these systems requires an accurate model of the impact of timing violations, and of the energy consumption. We propose to achieve this by characterizing a test circuit that is representative of the complete circuit implementation.

	The term \emph{deviation} refers to the effect of circuit faults on the result of a computation, and the deviation model is the bridge between the circuit characterization and the analysis of the algorithm. We reserve the term \emph{error} for describing the algorithm, in the present case to refer to the incorrect detection of a transmitted symbol.
	A timing violation occurs in the circuit when the propagation delay between the input and output extends beyond a clock period. Modeling the deviations introduced by timing violations is challenging because they not only depend on the current input to the circuit, but also on the state of the circuit before the new input was applied. In general, timing violations also depend on other dynamic factors and on process variations. 
	
	In this paper, we focus on the case where the output of the circuit is entirely determined by the current and previous inputs of the circuit, and by the nominal operating condition of the circuit. 
	We denote by $\Gamma$ the set of possible operating conditions, represented by vectors of parameters, and by $\gamma \in \Gamma$ a particular operating condition. For example, an operating condition might specify the supply voltage and clock period used in the circuit.
	We assume that all the parameters specified by $\gamma$ are deterministic.
	
	\subsection{Computation Model}\label{sec:deviation:proc}

	As described in Section~\ref{sec:architecture}, we consider a decoder composed of a processing unit (shown in Fig.~\ref{fig:layered_arch}) that computes all messages to and from one check node in each clock cycle.
	Since the circuit is synchronous, we can represent the computations in terms of a discrete-time system.
	Let $X_k$ be the input at clock cycle $k$. When timing violations are allowed to occur, the corresponding\footnote{The circuit could require one or several clock cycles to generate the first output, but this is irrelevant to the characterization of the computation.} circuit output $Z_k$ can be expressed as $Z_k = g(X_k, S_k)$, where $S_k$ represents the state of the circuit at the beginning of cycle $k$, and $g$ is some deterministic function.
	By definition, the state $S_k$ depends on the previous input $X_{k-1}$, but not on $X_k$.
	
	\begin{figure}[tbp]
	\begin{center}
	\subfloat[][Functional representation.]{\label{fig:LDPC_tree_dev_model}\includegraphics[width=2.7in]{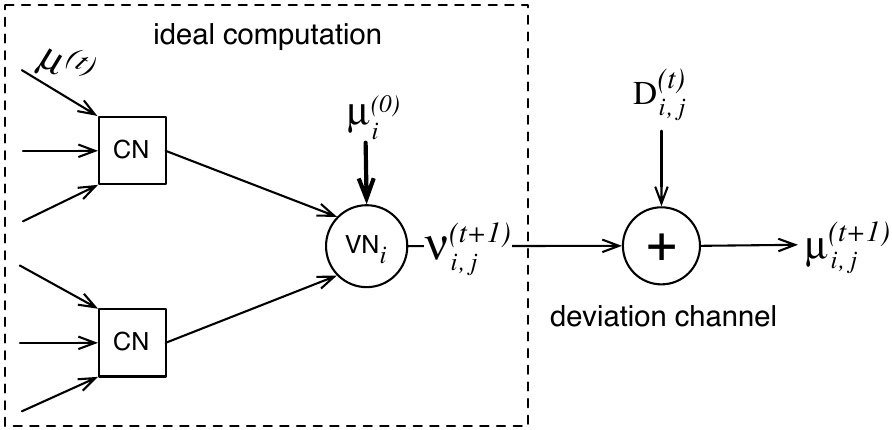}}
	\qquad
	\subfloat[][Bayesian representation.]{\label{fig:LDPC_tree_bayesian}\includegraphics[height=2in]{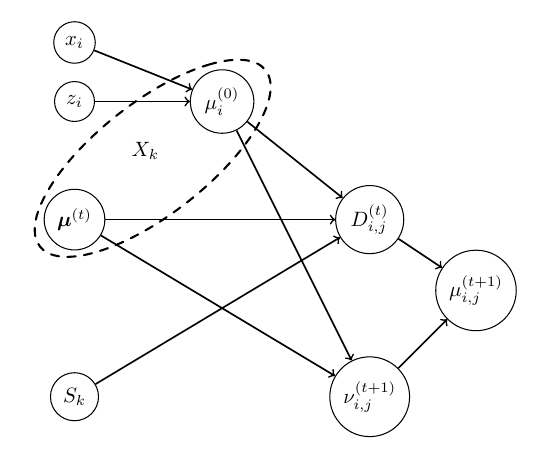}}
	\caption{Computation tree combined with the deviation model.}
	\label{fig:comptree_functionalbayesian}
	\end{center}
	\end{figure}
	
	The computations required to perform one decoding iteration can be represented using a computation tree, which models the generation of a VN-to-CN message in terms of messages $\mu^{(t)}$ sent in the previous iteration. There are $(d_v-1)$ check nodes in the tree. Each of these check nodes receives $(d_c-1)$ messages from neighboring variable nodes, and generates a message sent to the one VN whose message was excluded from the computation. This VN then generates an extrinsic message based on the channel prior $\mu_i^{(0)}$ and on the messages received from neighboring check nodes.
	An example of a computation tree is shown within the dashed box in Fig.~\ref{fig:LDPC_tree_dev_model}.
	In this paper, we assume that messages are updated using a flooding schedule, or in other words, that all messages $\mu^{(t)}$ at the left of the tree are identically distributed.

	Evaluating the computation tree requires $(d_v-1)$ uses of a processor, but the number of processors implemented in the decoder, and the way they are mapped to nodes of the Tanner graph can affect the modeling of deviations.
	Since a processor performs a parallel check node computation, $X_k$ and $Z_k$ are associated with $d_c$ distinct VNs. Let $\mathcal{V}_k$ denote the set of VNs associated with the processor during cycle $k$.
	Let us first assume that a processor is always mapped to distinct VNs in consecutive clock cycles, i.e.~for any processor, 
	\begin{equation}\label{eq:distinctVNs}
	\mathcal{V}_k \cap \mathcal{V}_{k-1} = \emptyset.
	\end{equation}
	Then, $X_k$ and $X_{k-1}$ are independent, and if they belong to the same decoding iteration, they are also identically distributed. As a result, $X_k$ and $S_k$ are also independent.
	At the output of the processor, $Z_k$ and $Z_{k-1}$ are not independent since they both depend on $X_{k-1}$. However, if \eqref{eq:distinctVNs} holds, messages received at a particular VN are guaranteed to have been generated in non-consecutive clock cycles, and it is therefore reasonable to consider only the marginal distribution of $Z_k$. 
	
	We now briefly describe some decoder architectures in which \eqref{eq:distinctVNs} holds.
	One possible architecture consists in implementing a single processor. Neglecting the circuit's latency, the processor would therefore be reused for $m/L$ cycles to compute each layer, and for $m$ cycles to perform one iteration.
	If we assume that the parity-check matrix contains at most one non-zero element per column and per layer, \eqref{eq:distinctVNs} clearly holds for cycles that belong to the same layer. Furthermore, since the order in which CNs are processed can be chosen arbitrarily, it can easily be chosen to ensure that \eqref{eq:distinctVNs} also holds when starting a new layer.
	Another architecture of particular interest is one that instantiates $m/L$ processors to achieve maximum parallelism. In this case, each processor is used once in each layer. This type of architecture is often used with quasi-cyclic parity-check matrices composed of cyclically shifted identity sub-matrices. In this case, \eqref{eq:distinctVNs} holds as long as the shift indices used in two consecutive layers are different. 
	
	For convenience, we choose to represent the effect of deviations at the output of the computation tree.
	The properties that have been established for the processors continue to apply when these processors are used to evaluate a computation tree, since this corresponds to the normal operation of the decoder that was just discussed. 
	Therefore, reusing the previous notation, we can write 
	$\mu_{i,j}^{(t+1)} = g(X_k, S_k)$, where $\mu_{i,j}^{(t+1)}$ is a VN-to-CN message that will be sent in the next iteration, and $S_k$ now represents the combined state of the processors that are used to evaluate the tree. 
	Defining $\bm{\mu}^{(t)}$ as a vector containing all the VN-to-CN messages that form the input of the computation tree, $X_k$ becomes the concatenation of $\bm{\mu}^{(t)}$ with $\mu_i^{(0)}$.

	\subsection{Deviation Model}\label{sec:deviation:model}
	
	We have seen above that a decoder message $\mu_{i,j}^{(t+1)}$ can be expressed as a function of the messages $\bm{\mu}^{(t)}$ received by the neighboring check nodes in the previous iteration and of the state of the processing circuit. To separate the deviations from the ideal operation of the decoder, it is helpful to decompose a decoding iteration into the ideal computation, followed by a transmission through a deviation channel.
	This model is shown in Fig.~\ref{fig:LDPC_tree_dev_model}, where $\nu_{i,j}^{(t+1)}$ is the message that would be sent from variable node $i$ to check node $j$ during iteration $t+1$ if no deviations had occurred during iteration $t$. 
	For the first messages sent in the decoder at $t=0$, the computation circuits are not used and therefore no deviation can occur, and we simply have $\mu_{i,j}^{(0)} = \nu_{i,j}^{(0)}$.
	Since we neglect correlations in successive circuit outputs, the deviation channel is memoryless.
	
	Unlike typical channel models where the noise is independent from other variables in the system,
	the deviation $D_{i,j}^{(t)}$ is a function of the current circuit input $X_k$
	and of the current state $S_k$. 
	Fig.~\ref{fig:LDPC_tree_bayesian} illustrates the dependencies between the random variables involved in the computation tree. Assuming a fixed distribution for $\bm{\mu}^{(t)}$ and $S_k$, $\phi(\mu_{i,j}^{(t+1)})$ can be determined by marginalizing the joint distribution of $x_i$, $z_i$, $\bm{\mu}^{(t)}$, $S_k$, and $\mu_{i,j}^{(t+1)}$, which in practice can be done using a Monte-Carlo simulation of a test circuit that implements the computation tree (such a simulation is discussed in more details in Appendix~\ref{sec:appendix:testcircuit}). However, directly measuring $\phi(\mu_{i,j}^{(t+1)})$ makes it difficult to extend the model to handle different input distributions. Instead, we propose a deviation model that retains conditional dependencies on the ideal message $\nu_{i,j}^{(t+1)}$ and on the transmitted bit $x_i$.
	We note that a similar model that did not include $x_i$ was also used in~\cite{dupraz:2015}.
	The deviation model is thus expressed as the conditional distribution $\phi(\mu_{i,j}^{(t+1)} \vert \nu_{i,j}^{(t+1)}, x_i)$.
	Using the ideal message as a model parameter rather than $X_k$ has the advantage of reducing the complexity of the model, but also of limiting the error when the model is used with other input distributions, as discussed in the following subsection.
	
	\subsection{Generalized Deviation Model}\label{sec:deviation:genmodel}
	The deviation model introduced in Section~\ref{sec:deviation:model} requires $\bm{\mu}^{(t)}$ and $S_k$ to have fixed distributions, but these distributions change at every decoding iteration.
	Furthermore, because the message distribution depends on the transmitted codeword, the deviation model also depends on the transmitted codeword.
	Under the assumption that $X_k$ and $X_{k-1}$ are \gls{iid}, $\pmf(S_k)$ is a function of $\pmf(\mu_{i,j}^{(t)})$, and it is thus sufficient to consider the evolution of $\pmf(\mu_{i,j}^{(t)})$.

	Let us first assume that the transmitted codeword is fixed. In this case, the message distribution $\pmf(\mu_{i,j}^{(t)})$ depends on the channel noise, on the iteration index $t$, and on the operating conditions of the circuit. 
	Since the messages are affected by deviations for $t>0$, only $\pmf(\mu_{i,j}^{(0)})$ is known a priori.
	An obvious way to measure deviations for all decoding iterations is to 
	determine $\phi(\mu_{i,j}^{(1)})$
	using the known $\pmf(\mu_{i,j}^{(0)})$, and to repeat the process for each subsequent decoding iteration. However, the resulting deviation model is of limited interest, since it depends on the specific message distributions in each iteration.
	
	To generate a model that is independent of the iterative progress of the decoder, we first approximate $\pmf(\mu_{i,j}^{(t)})$ as a \gls{1D} Normal distribution with error rate parameter $p_e^{(t)}$ chosen such that 
	\begin{equation}\label{eq:errorprob}
	p_e^{(t)} = \Pr(x_i \mu_{i,j}^{(t)} < 0) + \frac{1}{2}  \Pr(\mu_{i,j}^{(t)} = 0).
	\end{equation}
	This also provides us with an implicit parametrization of $\pmf(S_k)$ in terms of $p_e^{(t)}$.
	Note that while $\pmf(\mu_{i,j}^{(0)})$ does correspond exactly to a \gls{1D} Normal distribution, this is not necessarily the case after the first iteration. However, the approximation is only used to characterize deviations, and exact distributions can still be used to characterize messages. Therefore, $\pmf(\nu_{i,j}^{(t+1)})$ can be determined exactly, and the impact of the approximation remains small as long as $\Pr(\mu_{i,j}^{(t+1)} \neq \nu_{i,j}^{(t+1)})$ is small.
	In fact, combining a density evolution based on exact distributions with a deviation model generated using \gls{1D} Normal distributions leads to very accurate bit error rate predictions in practice \cite{leduc-primeau:2016a}.
	
	To evaluate $\phi(\mu_{i,j}^{(t)})$, we must also consider that deviations depend on the operating condition $\gamma$.
	Once $(p_e^{(t)}, \gamma)$ is specified, $\phi(\mu_{i,j}^{(t)})$ is uniquely determined by the synthesized circuit, and in order to retain the ability to represent arbitrary circuits, we make no assumption on the distribution and simply characterize it as an arbitrary conditional \gls{pmf}.
	We therefore obtain a model consisting of a family of non-parametric conditional \glspl{pmf} denoted as
	\begin{equation}\label{eq:devmodel2}
	\pmf^{(p_e^{(t)},\gamma)}\left(\mu_{i,j}^{(t+1)} \cond \nu_{i,j}^{(t+1)}, x_i \right),
	\end{equation}
	where $(p_e^{(t)}, \gamma)$ are the family parameters.
	However, we generally omit the $(p_e^{(t)},\gamma)$ superscript to simplify the notation.
	In practice, \eqref{eq:devmodel2} is constructed by performing several Monte-Carlo simulations of the circuit implementation of the computation tree in Fig.~\ref{fig:comptree_functionalbayesian} for various  $p_e^{(t)}$ values and for all operating conditions $\gamma \in \Gamma$. Interpolation is then used to obtain a continuous model in $p_e^{(t)}$.
	While measuring deviations, we also record the switching activity in the circuit, which is then used to construct an energy model that depends on $\gamma$ and $p_e^{(t)}$, denoted as $c_\gamma(p_e^{(t)})$ (where $c$ stands for ``cost'').
	
	To use the model, we first use \eqref{eq:errorprob} to determine the error rate parameter $p_e^{(t)}$ corresponding to the 
	arbitrary message distribution $\pmf(\mu_{i,j}^{(t)})$ at the beginning of the iteration, and we then
	 retrieve the appropriate conditional \gls{pmf} based on $p_e^{(t)}$ and on the operating condition $\gamma$.
	 This conditional \gls{pmf} then informs us of the statistics of deviations that occur at the end of the iteration, that is on messages sent in iteration $t+1$.
	
	As mentioned above, since $\pmf(\mu_{i,j}^{(t)})$ depends on the transmitted codeword, this is also the case of $\pmf(S_k)$ and of the deviation distributions.
	We show in Section~\ref{sec:analysis} that the codeword dependence is entirely contained within the deviation model and does not affect the analysis of the decoding performance, as long as the decoding algorithm and deviation model satisfy certain properties.
	Nonetheless, we would like to obtain a deviation model that does not depend on the transmitted codeword.
	This can be done when the objective is to predict the average performance of the decoder, rather than the performance for a particular codeword, since it is then sufficient to model the average behavior of the decoder.
	For the case where all codewords have an equal probability of being transmitted, we propose to perform the Monte-Carlo deviation measurements by randomly sampling transmitted codewords.
	This approach is supported by the experimental results presented in \cite{leduc-primeau:2016a}, which show that a deviation model constructed in this way can indeed accurately predict the average decoding performance.

\section{Performance Analysis}\label{sec:analysis}

	\subsection{Standard Analysis Methods for LDPC Decoders}\label{sec:std_ldpc_analysis}
	
	Density evolution (DE) is the most common tool used for predicting the error-correction performance of an LDPC decoder. The analysis relies on the assumption that messages passed in the Tanner graph are mutually independent, which holds as the code length goes to infinity \cite{richardson:2001a}.
	Given the channel output probability distribution and the probability distribution of variable node to check node messages at the start of an iteration, DE computes the updated distribution of variable node to check node messages at the end of the decoding iteration.
	This computation can be performed iteratively to determine the message distribution after any number of decoding iterations.
	The validity of the analysis rests on two properties of the LDPC decoder. The first property is the conditional independence of errors, which states that the error-correction performance of the decoder is independent from the particular codeword that was transmitted.
	The second property states that the error-correction performance of a particular LDPC code concentrates around the performance measured on a cycle-free graph, as the code length goes to infinity.
	
	Both properties were shown to hold in the context of reliable implementations \cite{richardson:2001a}.
	It was also shown that the conditional independence of errors always holds when the channel is output symmetric 
	and the decoder has a symmetry property. We can define a sufficient symmetry property of the decoder in terms of a message-update function $F_{i,j}$ that represents one complete iteration of the (ideal) decoding algorithm. Given a vector of all the messages $\bvec{\mu}^{(t)}$ sent from variable nodes to check nodes at the start of iteration $t$ and the channel information $\nu_i^{(0)}$ associated with variable node $i$, $F_{i,j}$ returns the next ideal message to be sent from a variable node $i$ to a check node $j$: 
$\nu_{i,j}^{(t+1)}=F_{i,j}\left(\bvec{\mu}^{(t)}, \nu_i^{(0)}\right)$. 
	\begin{definition}\label{def:decoder_symmetry}
	A message-update function $F_{i,j}$ is said to be \emph{symmetric} with respect to a code $C$ if 
	\[
	F_{i,j}\left(\bvec{\mu}^{(t)}, \nu_i^{(0)}\right) = x_i F_{i,j}\left(\bvec{x} \bvec{\mu}^{(t)}, x_i \nu_i^{(0)}\right)
	\]
	for any $\bvec{\mu}^{(t)}$, any $\nu_i^{(0)}$, and any codeword $\bvec{x} \in C$.
	\end{definition}%
	In other words, a decoder's message-update function is symmetric if multiplying all the VN-to-CN belief messages sent at iteration $t$ and the belief priors by a valid codeword $\bvec{x} \in C$ is equivalent to multiplying the next messages sent at iteration $t+1$ by that same codeword.
	Note that the symmetry condition in Definition~\ref{def:decoder_symmetry} is implied by the check node and variable node symmetry conditions in \cite[Def.~1]{richardson:2001a}. 

	\subsection{Applicability of Density Evolution}
	\label{sec:applicabilityDE}
	
	In order to use density evolution to predict the performance of long finite-length codes, the decoder must satisfy the two properties stated in Section~\ref{sec:std_ldpc_analysis}, namely the conditional independence of errors and the convergence to the cycle-free case.
	We first present some properties of the decoding algorithm and of the deviation model that are sufficient to ensure the conditional independence of errors.
	
	Using the multiplicative description of the BIAWGN channel, the vector received by the decoder is given by $\bm{y}=\bm{xz}$ when a codeword $\bm{x}$ is transmitted, or by $\bm{y}=\bm{z}$ when the all-one codeword is transmitted.
	In a reliable decoder, messages are completely determined by the received vector, but in a faulty decoder, there is additional randomness that results from the deviations. Therefore, we represent messages in terms of conditional probability distributions given $\bvec{xz}$.
	Since we are concerned with a fixed-point circuit implementation of the decoder, we can assume that messages are integers from the set $\{-Q, -Q+1, \dots, Q\}$, where $Q>0$ is the largest message magnitude that can be represented.
	\begin{definition}
	We say that a message distribution $\pmf_{\mu_{i,j}\vert\bvec{y}}(\mu \vert \bvec{xz})$ is symmetric if
	\[
	\pmf_{\mu_{i,j}\vert\bvec{y}}\left(\mu \cond \bvec{xz}\right) = 
	\pmf_{\mu_{i,j}\vert\bvec{y}}\left(x_i \mu \cond \bvec{z}\right) \, .
	\]
	\end{definition}	
	If a message has a symmetric distribution, its error probability as defined in \eqref{eq:errorprob} is the same whether $\bvec{xz}$ or $\bvec{z}$ is received.
	Similarly to the results presented in \cite{dupraz:2015}, we can show that the symmetry of message distributions is preserved when the message-update function is symmetric.
	\begin{lemma}\label{lem:ideal}
	If $F_{i,j}$ is a symmetric message-update function and if $\mu_i^{(0)}$ and $\mu_{i,j}^{(t)}$ have symmetric distributions for all $(i,j)$, the next ideal messages $\nu_{i,j}^{(t+1)}$ also have symmetric distributions.
	\end{lemma}
	\begin{proof}
	We can express the distribution of the next ideal message from VN $i$ to CN $j$ 
	as
	\ifCLASSOPTIONdraftcls 
	\begin{equation}\label{eq:nextideal}
	\pmf_{\nu_{i,j}^{(t+1)}\vert\bvec{y}}\left(\nu \cond \bvec{xz}\right) = 
		\sum_{(\bvec{\mu}, \mu_i^{(0)}) \in R} \pmf_{\bvec{\mu}^{(t)} \vert \bvec{y}}\!\left(\bvec{\mu} \cond \bvec{xz}\right) \,
								\pmf_{\mu_i^{(0)}\vert\bvec{y}}\!\left(\mu_i^{(0)} \cond \bvec{xz}\right),
	\end{equation}
	\else 
	\begin{multline}\label{eq:nextideal}
	\pmf_{\nu_{i,j}^{(t+1)}\vert\bvec{y}}\left(\nu \cond \bvec{xz}\right) = \\
		\sum_{(\bvec{\mu}, \mu_i^{(0)}) \in R} \pmf_{\bvec{\mu}^{(t)} \vert \bvec{y}}\!\left(\bvec{\mu} \cond \bvec{xz}\right) \,
								\pmf_{\mu_i^{(0)}\vert\bvec{y}}\!\left(\mu_i^{(0)} \cond \bvec{xz}\right),	
	\end{multline}
	\fi
	where $R= \left\{ (\bvec{\mu},\mu_i^{(0)}) : F_{i,j}\left(\bvec{\mu}, \mu_i^{(0)}\right) = \nu \right\}$.
	
	Assuming that the elements of the VN-to-CN message vector $\bm{\mu}^{(t)}$ are independent and that each $\mu_{i,j}^{(t)}$ has a symmetric distribution,
	\ifCLASSOPTIONdraftcls 
	\[
	\pmf_{\bvec{\mu}^{(t)} \vert \bvec{y}}\!\left(\bvec{\mu} \cond \bvec{xz}\right)
	= \prod_k \pmf_{\mu_k^{(t)} \vert \bvec{y}}\!\left(\mu_k \cond \bvec{xz}\right)
	= \prod_k \pmf_{\mu_k^{(t)} \vert \bvec{y}}\!\left(x_k \mu_k \cond \bvec{z}\right)
	= \pmf_{\bvec{\mu}^{(t)} \vert \bvec{y}}\!\left(\bvec{x\mu} \cond \bvec{z}\right),
	\]
	\else 
	\begin{align*}
	\pmf_{\bvec{\mu}^{(t)} \vert \bvec{y}}\!\left(\bvec{\mu} \cond \bvec{xz}\right)
	&= \prod_k \pmf_{\mu_k^{(t)} \vert \bvec{y}}\!\left(\mu_k \cond \bvec{xz}\right) \\
	&= \prod_k \pmf_{\mu_k^{(t)} \vert \bvec{y}}\!\left(x_k \mu_k \cond \bvec{z}\right) \\
	&= \pmf_{\bvec{\mu}^{(t)} \vert \bvec{y}}\!\left(\bvec{x\mu} \cond \bvec{z}\right),	
	\end{align*}
	\fi
	and since the channel output $\mu_i^{(0)}$ also has a symmetric distribution,
	\[
	\pmf_{\mu_i^{(0)}\vert\bvec{y}}\!\left(\mu_i^{(0)} \cond \bvec{xz}\right)
	= \pmf_{\mu_i^{(0)}\vert\bvec{y}}\!\left(x_i \mu_i^{(0)} \cond \bvec{z}\right).
	\]
	Therefore, we can rewrite \eqref{eq:nextideal} as
	\ifCLASSOPTIONdraftcls 
	\begin{equation}\label{eq:nextideal2}
	\pmf_{\nu_{i,j}^{(t+1)}\vert\bvec{y}}\left(\nu \cond \bvec{xz}\right) = 
		\sum_{(\bvec{\mu}, \mu_i^{(0)}) \in R} \pmf_{\bvec{\mu}^{(t)} \vert \bvec{y}}\!\left(\bvec{x\mu} \cond \bvec{z}\right) \,
								\pmf_{\mu_i^{(0)}\vert\bvec{y}}\!\left(x_i \mu_i^{(0)} \cond \bvec{z}\right).
	\end{equation}
	\else 
	\begin{multline}\label{eq:nextideal2}
	\pmf_{\nu_{i,j}^{(t+1)}\vert\bvec{y}}\left(\nu \cond \bvec{xz}\right) = \\
		\sum_{(\bvec{\mu}, \mu_i^{(0)}) \in R} \pmf_{\bvec{\mu}^{(t)} \vert \bvec{y}}\!\left(\bvec{x\mu} \cond \bvec{z}\right) \,
								\pmf_{\mu_i^{(0)}\vert\bvec{y}}\!\left(x_i \mu_i^{(0)} \cond \bvec{z}\right).	
	\end{multline}
	\fi
	Finally, letting $\bvec{\mu}'=\bvec{x\mu^{(t)}}$ and $\nu_i' = x_i \mu_i^{(0)}$, \eqref{eq:nextideal2} becomes
	\[
	\pmf_{\nu_{i,j}^{(t+1)}\vert\bvec{y}}\left(\nu \cond \bvec{xz}\right) = 
		\sum_{(\bvec{\mu}', \nu_i') \in R'} \pmf_{\bvec{\mu}^{(t)} \vert \bvec{y}}\!\left(\bvec{\mu'} \cond \bvec{z}\right) \,
								\pmf_{\mu_i^{(0)}\vert\bvec{y}}\!\left(\nu_i' \cond \bvec{z}\right),
	\]
	where $R' = \left\{ (\bvec{\mu}', \nu_i') : F_{i,j}(\bvec{x\mu'}, x_i \nu_i') = \nu \right\}$.
	Since $F_{i,j}$ is symmetric, we can also express $R'$ as
	\[
	R' = \left\{ (\bvec{\mu}', \nu_i') : F_{i,j}(\bvec{\mu'}, \nu_i') = x_i \nu \right\},
	\]
	and therefore,
	\begin{align*}
	\pmf_{\nu_{i,j}^{(t+1)}\vert\bvec{y}}\left(x_i \nu \cond \bvec{z}\right)
		&= \sum_{(\bvec{\mu}', \nu_i') \in R'} \pmf_{\bvec{\mu}^{(t)} \vert \bvec{y}}\!\left(\bvec{\mu'} \cond \bvec{z}\right) \,
								\pmf_{\mu_i^{(0)}\vert\bvec{y}}\!\left(\nu_i' \cond \bvec{z}\right) \\
		&= \pmf_{\nu_{i,j}^{(t+1)}\vert\bvec{y}}\left(\nu \cond \bvec{xz}\right),
	\end{align*}
	indicating that the next ideal messages have symmetric distributions.
	\end{proof}
	
	To establish the conditional independence of errors under the proposed deviation model, we first define some properties of the deviation.
	\begin{definition}\label{def:devsym}
	We say that the deviation model is \emph{symmetric} if 
	\ifCLASSOPTIONdraftcls 
	\[ 
	\pmf_{\mu_{i,j}^{(t)} \vert \nu_{i,j}^{(t)}, \bvec{y}}\left(\mu \cond \nu, \bvec{xz}\right)
	= \pmf_{\mu_{i,j}^{(t)} \vert \nu_{i,j}^{(t)}, \bvec{y}}\left(\mu \cond \nu, \bvec{z}\right) 
	= \pmf_{\mu_{i,j}^{(t)} \vert \nu_{i,j}^{(t)}, \bvec{y}}\left(-\mu \cond -\nu, \bvec{z}\right).
	\]
	\else
	\begin{align*}
	\pmf_{\mu_{i,j}^{(t)} \vert \nu_{i,j}^{(t)}, \bvec{y}}\left(\mu \cond \nu, \bvec{xz}\right)
	&= \pmf_{\mu_{i,j}^{(t)} \vert \nu_{i,j}^{(t)}, \bvec{y}}\left(\mu \cond \nu, \bvec{z}\right) \\
	&= \pmf_{\mu_{i,j}^{(t)} \vert \nu_{i,j}^{(t)}, \bvec{y}}\left(-\mu \cond -\nu, \bvec{z}\right).	
	\end{align*}
	\fi
	\end{definition}
	\begin{definition}\label{def:weaksym}
	We say that the deviation model is \emph{weakly symmetric (WS)} if
	\[ 
	\pmf_{\mu_{i,j}^{(t)} \vert \nu_{i,j}^{(t)}, \bvec{y}}\left(\mu \cond \nu, \bvec{xz}\right)
	= \pmf_{\mu_{i,j}^{(t)} \vert \nu_{i,j}^{(t)}, \bvec{y}}\left(x_i \mu \cond x_i \nu, \bvec{z}\right).
	\]
	\end{definition}
	Note that if the model satisfies the symmetry condition, it also satisfies the weak symmetry condition, since $x_i \in \{-1, 1\}$. 
	We then have the following Lemma.
	\begin{lemma}\label{lem:independencetx}
	If a decoder having a symmetric message-update function and taking its inputs from an output-symmetric communication channel is affected by weakly symmetric deviations, its message error probability at any iteration $t\geq0$ is independent of the transmitted codeword.
	\end{lemma}
	\begin{proof}
	Similarly to the approach used in \cite[Lemma~4.90]{richardson:2008} and \cite{varshney:2011}, we want to show that the probability that messages are in error is the same whether $\bvec{xz}$ or $\bvec{z}$ is received. This is the case if the faulty messages $\mu_{i,j}^{(t)}$ have a symmetric distribution for all $t\geq0$ and all $(i,j)$.
	
	Since the communication channel is output symmetric and since no deviations can occur before the first iteration, messages $\mu_{i,j}^{(0)} = \nu_{i,j}^{(0)}$ have a symmetric distribution.
	We proceed by induction to establish the symmetry of the messages for $t>0$. We start by assuming that 
	\begin{equation}\label{eq:induct_assumpt}
	\pmf_{\nu_{i,j}^{(t)}\vert\bvec{y}}\left(\nu \cond \bvec{xz}\right) = 
	\pmf_{\nu_{i,j}^{(t)}\vert\bvec{y}}\left(x_i \nu \cond \bvec{z}\right)
	\end{equation}
	also holds for $t>0$.

	Using Definition~\ref{def:weaksym} and \eqref{eq:induct_assumpt}, we can write the faulty message distribution as
	\begin{align*}
	\pmf_{\mu_{i,j}^{(t)} \vert \bvec{y}}\left(\mu \cond \bvec{xz} \right)
		&= \sum_{\nu=-Q}^{Q} \pmf_{\mu_{i,j}^{(t)} \vert \bvec{y}}\left(\mu \cond \nu, \bvec{xz}\right) 
						   \pmf_{\nu_{i,j}^{(t)} \vert \bvec{y}}\left(\nu \cond \bvec{xz}\right) \\
		&= \sum_{\nu=-Q}^{Q} \pmf_{\mu_{i,j}^{(t)} \vert \bvec{y}}\left(x_i \mu \cond x_i \nu, \bvec{z}\right) 
						   \pmf_{\nu_{i,j}^{(t)} \vert \bvec{y}}\left(x_i \nu \cond \bvec{z}\right) \\
		&= \sum_{\nu'=-x_i Q}^{x_i Q} \pmf_{\mu_{i,j}^{(t)} \vert \bvec{y}}\left(x_i \mu \cond \nu', \bvec{z}\right) 
						   \pmf_{\nu_{i,j}^{(t)} \vert \bvec{y}}\left(\nu' \cond \bvec{z}\right) \\
		&= \sum_{\nu'=-Q}^{Q} \pmf_{\mu_{i,j}^{(t)} \vert \bvec{y}}\left(x_i \mu \cond \nu', \bvec{z}\right) 
						   \pmf_{\nu_{i,j}^{(t)} \vert \bvec{y}}\left(\nu' \cond \bvec{z}\right) \\
		&= \pmf_{\mu_{i,j}^{(t)} \vert \bvec{y}}\left(x_i \mu \cond \bvec{z}\right).
	\end{align*}
	where the third equality is obtained using the substitution $\nu'=x_i\nu$.
	We conclude that the faulty messages have a symmetric distribution.
	Finally, since the decoder's message-update function is symmetric, Lemma~\ref{lem:ideal} confirms the induction hypothesis in \eqref{eq:induct_assumpt}.	
	\end{proof}

	The last remaining step in establishing whether density evolution can be used with a decoder affected by WS deviations is to determine whether the error-correction performance of a code concentrates around the cycle-free case.
	The property has been shown to hold in \cite{varshney:2011} (Theorems 2, 3 and 4) for an LDPC decoder affected by ``wire noise'' and ``computation noise''. The wire noise model is similar to our deviation model, in the sense that the messages are passed through an additive noise channel, and that the noise applied to one message is independent of the noise applied to other messages. The proof presented in \cite{varshney:2011} only relies on the fact that the wire noise applied to a given message can only affect messages that are included in the directed neighborhood of the edge where it is applied, where the graph direction refers to the direction of message propagation. This clearly also holds in the case of our deviation model, and therefore the proof is the same.

	Since the message error probability is independent of the transmitted codeword, and furthermore concentrates around the cycle-free case, density evolution can be used to determine the error-correction performance of a decoder perturbed by our deviation model, as long as the deviations are weakly symmetric.
	It is important to note that as discussed in Section~\ref{sec:deviation:genmodel}, the deviation model itself still depends on the transmitted codeword. However, given a weakly symmetric deviation model, density evolution can be used to determine the decoder's performance. The hope is that in practice, only a single (or a few) deviation models are required to represent the deviations for all codewords, and indeed one model is sufficient to obtain accurate predictions in the experiment of~\cite{leduc-primeau:2016a}.

	\subsection{Deviation Examples}\label{sec:devexamples}
	As described in Section~\ref{sec:deviation:genmodel}, we collect deviation measurements from the test circuits by inputting test vectors representing random codewords, and distributed according to several $p_e^{(t-1)}$ values. We then generate estimates of the conditional \glspl{pmf} in \eqref{eq:devmodel2}.
	It is interesting to visualize the distributions using an aggregate measure such as the probability of observing a non-zero deviation 
	\begin{equation}\label{eq:pnz}
	p_\mathrm{nz}(\nu_{i,j}^{(t)},x_i) = \Pr^{(p_e^{(t-1)},\gamma)}\left(\mu_{i,j}^{(t)} \neq \nu_{i,j}^{(t)} \cond \nu_{i,j}^{(t)}, x_i \right). 
	\end{equation}
	These conditional probabilities are shown for a $(3,30)$ circuit in Fig.~\ref{fig:devPr_3-30}. When $x_i=1$, positive belief values indicate a correct decision, whereas when $x_i=-1$, negative belief values indicate a correct decision. 
	We can see that in this example, deviations are more likely when the belief is incorrect than when it is correct, and therefore a symmetric deviation model is not consistent with these measurements.
	On the other hand, there is a sign symmetry between the ``correct'' part of the curves, and between the ``incorrect'' parts, 
	that is $p_\mathrm{nz}(\nu_{i,j}^{(t)},1) \approx p_\mathrm{nz}(-\nu_{i,j}^{(t)},-1)$,
	and for this reason a weakly symmetric model is consistent with the measurements.
	Note that the slight jaggedness observed for incorrect belief values of large magnitude in the $p_e^{(t-1)}=0.008$ curves is due to the fact that these $\nu_{i,j}$ values occur only rarely. For the largest incorrect $\nu_{i,j}$ values, only about 100 deviation events are observed for each point, despite the large number of \gls{MC} trials. 

\begin{figure}[tbp]
\centering
\includegraphics[width=2.8in]{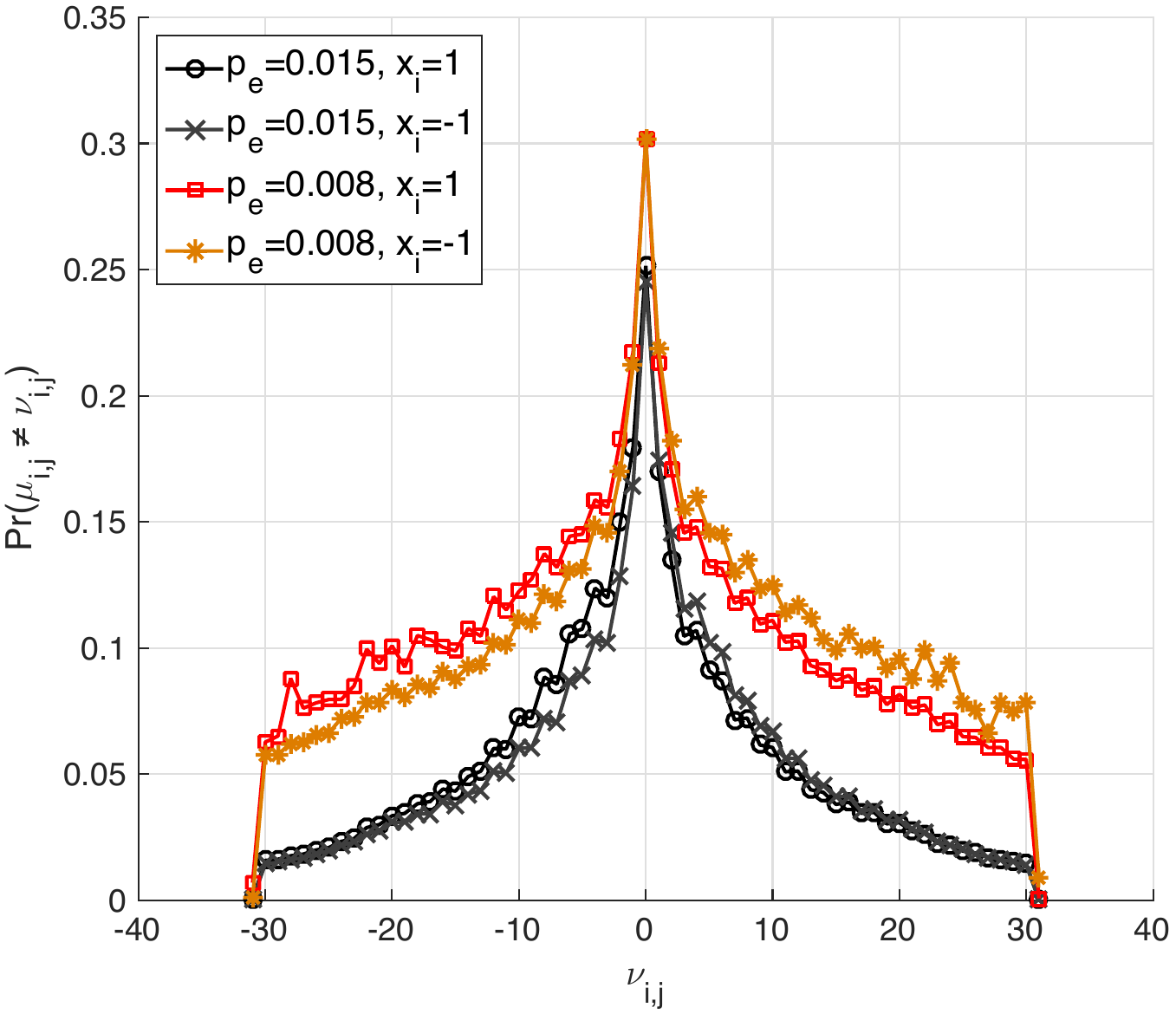} 
\caption{Non-zero deviation probability given $\nu_{i,j}^{(t)}$ and $x_i$ at two $p_e^{(t-1)}$ values, measured on a $(3,30)$ circuit operated at $\vdd=0.75\volt$ and $T_\mathrm{clk}=3.2\ns$. $3 \cdot 10^8$ decoding iteration trials were performed for each $p_e^{(t-1)}$ value. The total number of non-zero deviation events observed is 4,115,229 at $p_e^{(t-1)}=0.015$, and 10,071,810 at $p_e^{(t-1)}=0.008$.}
\label{fig:devPr_3-30}
\end{figure}

	Figure~\ref{fig:devPr_3-6} shows a similar plot for a $(3,6)$ circuit. In this case, $p_\mathrm{nz}(\nu_{i,j}^{(t)},x_i) \approx p_\mathrm{nz}(-\nu_{i,j}^{(t)},x_i)$, and a symmetric deviation model could be appropriate. Of course, since it is more general, a WS model is also appropriate.

\begin{figure}[tbp]
\centering
\includegraphics[width=2.8in]{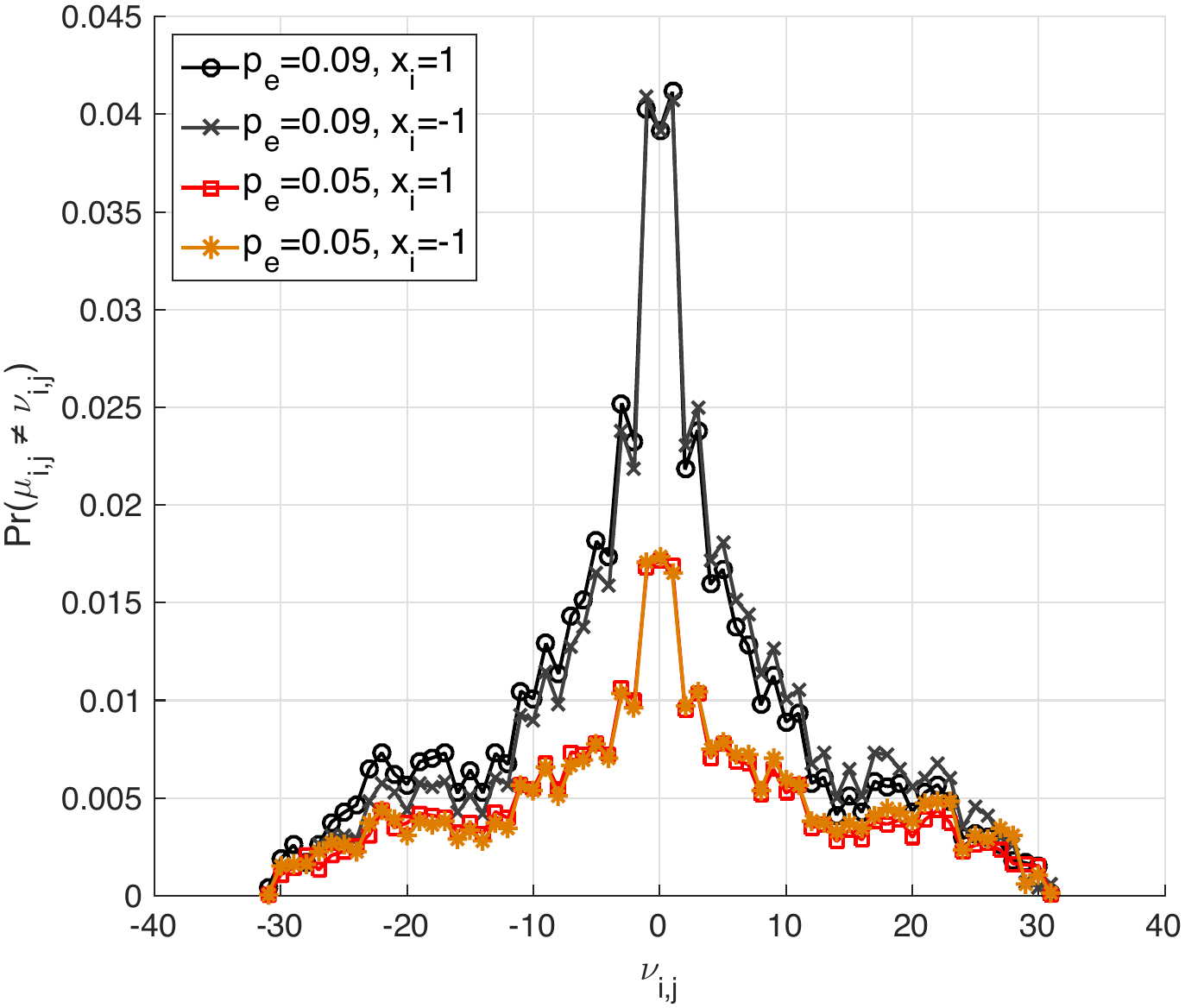} 
\caption{Non-zero deviation probability given $\nu_{i,j}^{(t)}$ and $x_i$ at two $p_e^{(t-1)}$ values, measured on a $(3,6)$ circuit operated at $\vdd=0.85\volt$ and $T_\mathrm{clk}=2.1\ns$. $3 \cdot 10^8$ decoding iteration trials were performed for each $p_e^{(t-1)}$ value. The total number of non-zero deviation events observed is 2,524,601 at $p_e^{(t-1)}=0.09$, and 1,020,867 at $p_e^{(t-1)}=0.05$.}
\label{fig:devPr_3-6}
\end{figure}

	Under the assumption that deviations are weakly symmetric, we have
	\[
	\pmf_{\mu_{i,j}^{(t)} | \nu_{i,j}^{(t)}, x_i}\left(\mu \cond \nu, 1\right) =
	\pmf_{\mu_{i,j}^{(t)} | \nu_{i,j}^{(t)}, x_i}\left(-\mu \cond -\nu, -1\right).
	\]
	Therefore, we can combine the $x_i=1$ and $x_i=-1$ data generated by the \gls{MC} simulation to improve the accuracy of the estimated \glspl{pmf}.
	To determine the validity of the WS assumption in a systematic way, we can generate an error metric by applying the WS assumption to one half of the simulation data to predict the other half. For all the circuits and operating conditions considered, the mean squared error of the predicted \glspl{pmf} remains below $10^{-3}$.
		
	Let $p_L$ and $p_H$ be respectively the smallest and largest $p_e^{(t-1)}$ values for which the deviations have been characterized. We can generate a conditional \gls{pmf} for any $p_e^{(t-1)} \in [p_L, p_H]$ by interpolating from the nearest \glspl{pmf} that have been measured. We choose $p_H \geq p_e^{(0)}$ to make sure that the first iteration's deviation is within the characterized range. 
	Because messages in the decoder are saturated once they reach the largest magnitude that can be represented, and since messages are represented in the CNP in sign \& magnitude format, the circuit's switching activity decreases when the message error probability becomes very small. Since timing faults cannot occur when the circuit does not switch, we can expect deviations to be equally or less likely at $p_e^{(t-1)}$ values below $p_L$.
	Therefore, to define the deviation model for $p_e^{(t-1)}<p_L$, we make the pessimistic assumption that the deviation \gls{pmf} remains the same as for $p_e^{(t-1)}=p_L$.

	\subsection{DE and Energy Curves}\label{sec:DE_energ_curves}
	We evaluate the progress of the decoder affected by timing violations using quantized density evolution \cite{chung:2001}. For the Offset Min-Sum algorithm, a DE iteration can be split into the following steps: 1-a)~evaluating the distribution of the CN minimum, 1-b)~evaluating the distribution of the CN output, after subtracting the offset, 2)~evaluating the distribution of the ideal VN-to-CN message, and 3)~evaluating the distribution of the faulty VN-to-CN messages. Step 1-a is given in \cite{balatsoukas-stimming:2014}, while the others are straightforward.
	In the context of DE, we write the message distribution as $\bvec{\pi}^{(t)}= \pmf(\mu_{i,j}^{(t)} | x_i=1)$, and the channel output distribution as $\bvec{\pi}^{(0)}= \pmf(\mu_i^{(0)} | x_i=1)$.
	We write a DE iteration as $\bvec{\pi}^{(t+1)} = f_\gamma(\bvec{\pi}^{(t)}, \bvec{\pi}^{(0)})$.

	As mentioned in Section~\ref{sec:deviation:genmodel}, the energy consumption is modeled in terms of the message error probability and of the operating condition, and denoted $c_\gamma(p_e^{(t)})$.
	As for the deviation model, we use interpolation to define $c_\gamma(p_e^{(t)})$ for $p_e^{(t)} \in [p_L, p_H]$, and assume that $c_\gamma(p_e^{(t)})=c_\gamma(p_L)$ for $p_e^{(t)} < p_L$.
	To display $f_\gamma(\bvec{\pi}^{(t)},\bvec{\pi}^{(0)})$ and $c_\gamma(p_e^{(t)})$ on the same plot, we project $\bvec{\pi}^{(t)}$ onto the message error probability space.
	
\begin{figure}[tbp]
\begin{center}
\includegraphics[width=2.915in]{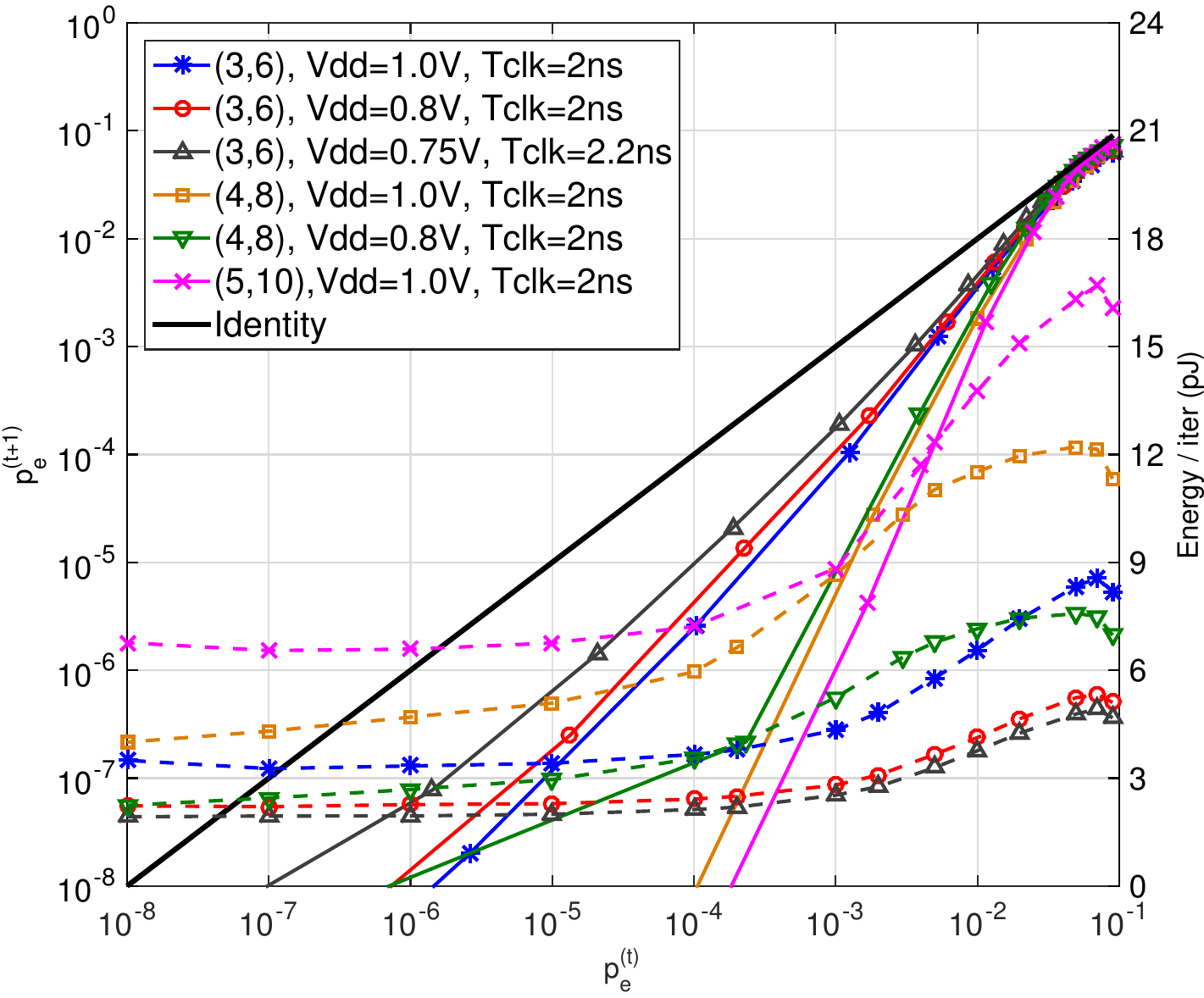} 
\caption{Examples of projected DE curves (solid lines) and energy curves (dashed lines) for rate $0.5$ ensembles with $d_v\in\{3,4,5\}$, and $p_e^{(0)}=0.09$.}
\label{fig:EXITexample}
\end{center}
\end{figure}

\begin{figure}[tbp]
\begin{center}
\includegraphics[width=2.915in]{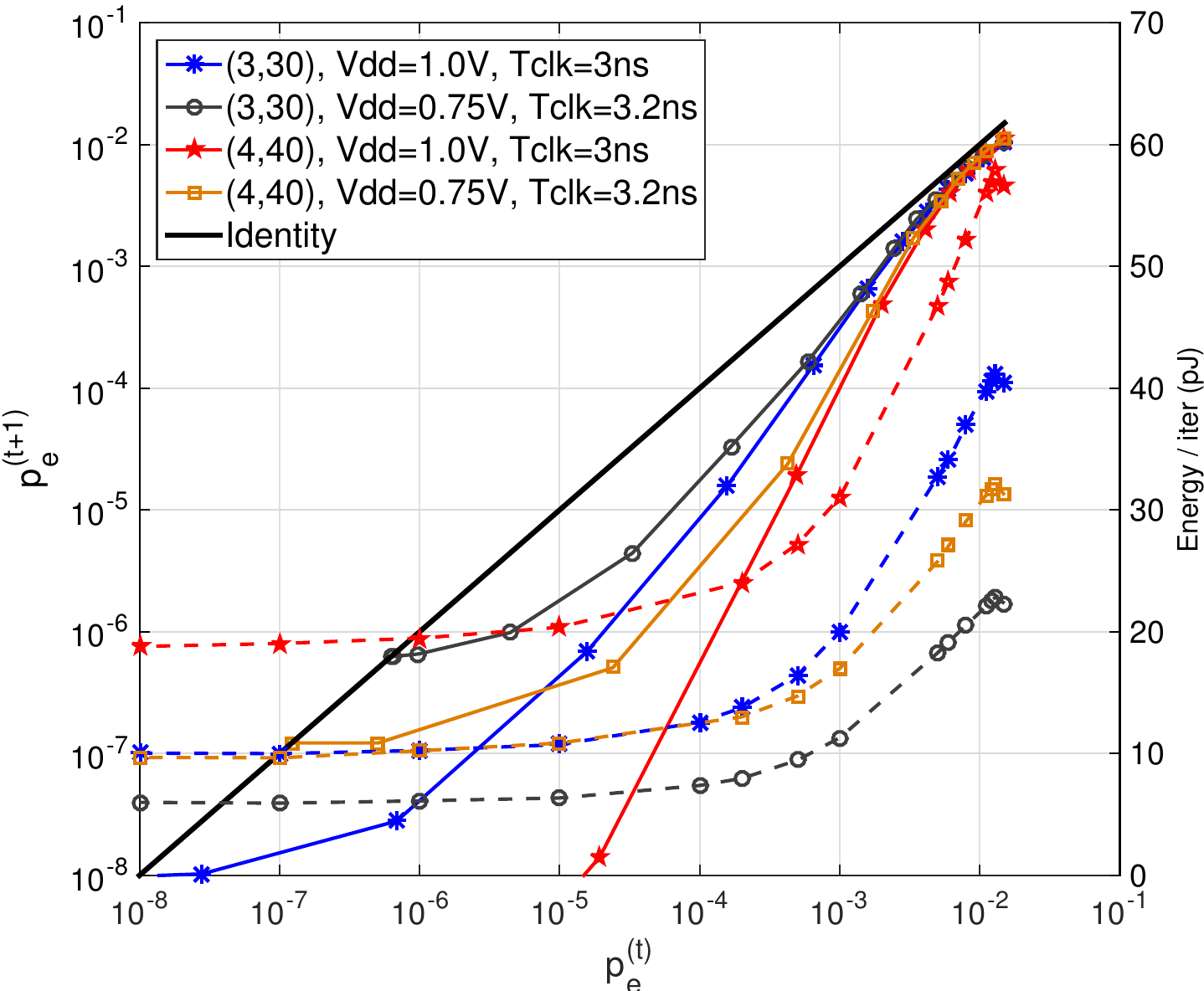} 
\caption{Examples of projected DE curves (solid lines) and energy curves (dashed lines) for the $(3,30)$ and $(4,40)$ ensembles (rate $0.9$), with $p_e^{(0)}=0.015$.}
\label{fig:EXITexample2}
\end{center}
\end{figure}

	Several regular code ensembles were evaluated, with rates $\frac{1}{2}$ and $\frac{9}{10}$.	
	Fig.~\ref{fig:EXITexample} shows examples of projected DE curves and energy curves for rate-$\frac{1}{2}$ code ensembles with $d_v \in \{3,4,5\}$ and various operating conditions. The energy is measured as described in Appendix~\ref{sec:appendix:workflow} and corresponds to one use of the test circuit (shown in Fig.~\ref{fig:testcircuit}).
	The nominal operating condition is $\vdd=1.0\volt$, $T_\mathrm{clk}=2.0\ns$ and therefore these curves correspond to a reliable implementation. With a reliable implementation, these ensembles have a channel threshold of $p_e^{(0)}\leq 0.12$ for the $(3,6)$ ensemble, $p_e^{(0)}\leq 0.11$ for $(4,8)$, and $p_e^{(0)}\leq 0.09$ for $(5,10)$. 
	We use $p_e^{(0)} = 0.09$ for all the curves shown in Fig.~\ref{fig:EXITexample} to allow comparing the ensembles. As can be expected, a larger variable node degree results in faster convergence towards zero error rate, and it is natural to ask whether this property might provide greater fault tolerance and ultimately better energy efficiency. This is discussed in Section~\ref{sec:results}.
	
	Fig.~\ref{fig:EXITexample2} is a similar plot for the $(3,30)$ and $(4,40)$ ensembles. The channel threshold of both ensembles is approximately $p_e^{(0)}\leq 0.019$. For these curves, the nominal operating condition is $\vdd=1.0\volt$ and $T_\mathrm{clk}=3\ns$.
	As we can see, the energy consumption per iteration of the $(4,40)$ decoder is roughly double that of the $(3,30)$ decoder.
	We note that in the case of the $(3,30)$ ensemble, the reliable decoder stops making progress at an error probability of approximately $10^{-8}$. This floor is the result of the message saturation limit chosen for the circuit.

\section{Energy Optimization}\label{sec:optimization}

	\subsection{Design Parameters}\label{sec:LDPC_param}

	As in a standard LDPC code-decoder design, the first parameter to be optimized is the choice of code ensemble. In this paper we restrict the discussion to regular codes, and therefore we need only to choose a degree pair $(d_v, d_c)$, where $R=1-d_v/d_c$ is the design rate of the code.
	For a fixed $R$, we can observe that both the energy consumption and the circuit area of the decoding circuit grow rapidly with $d_v$, and therefore it is only necessary to consider a few of the lowest $d_v$ values.

	Besides the choice of ensemble, we are interested in finding the optimal choice of operating parameters for the quasi-synchronous circuit. We consider here the supply voltage ($\vdd$) and the clock period ($T_\mathrm{clk}$).
	Generally speaking, the supply voltage affects the energy consumption, while the clock period affects the decoding time, or latency. The energy and latency are also affected by the choice of code ensemble, since the number of operations to be performed depends on the node degrees.
	The operating parameters of a decoder are denoted as a vector $\gamma=[\vdd,T_\mathrm{clk}]$. 
	
	The decoding of LDPC codes proceeds in an iterative fashion, and it is therefore possible to adjust the operating parameters on an iteration-by-iteration basis.
	In practice, this could be implemented in various ways, for example by using a pipelined sequence of decoder circuits, where each decoder is responsible for only a portion of the decoding iterations. 
	It is also possible to rapidly vary the clock frequency of a given circuit by using a digital clock divider circuit \cite{fischer:2006}.
	We denote by $\gammaseq$ the sequence of parameters used at each iteration throughout the decoding, and we use $\gammaseq=[\gamma_1^{N_1}, \gamma_2^{N_2}, \dots]$ to denote a specific sequence in which the parameter vector $\gamma_1$ is used for the first $N_1$ iterations, followed by $\gamma_2$ for the next $N_2$ iterations, and so on.

	\subsection{Objective}\label{sec:optimization:obj}
	The performance of the LDPC code and of its decoder can be described by specifying a vector $\bvec{P}=(p_e^{(0)},\pres,T_\mathrm{dec})$, where
	$p_e^{(0)}$ is the output error rate of the communication channel, 
	$\pres$ the residual error rate of VN-to-CN messages when the decoder terminates,
	and $T_\mathrm{dec}$ the expected decoding latency.
	
	The decoder's performance $\bvec{P}$ and energy consumption $E$ are controlled by $\gammaseq$.
	The energy minimization problem can be stated as follows. Given a performance constraint $\bvec{P}=(a,b,c)$, we wish to find the value of $\gammaseq$ that minimizes $E$, subject to $p_e^{(0)} = a$, $\pres \leq b$, $T_\mathrm{dec} \leq c$. 
	As in the standard DE method, we propose to use the code's computation tree as a proxy for the entire decoder, and furthermore to use the energy consumption of the test circuit described in Appendix~\ref{sec:appendix:testcircuit} as the optimization objective.
	To be able to replace the energy minimization of the complete decoder with the energy minimization of the test circuit, we make the following assumptions:
	\begin{enumerate}
	\item The ordering of the energy consumption is the same for the test circuit and for the complete decoder, that is, for any $\gamma_1$ and $\gamma_2$, $E_\textsc{test}(\gamma_1) \leq E_\textsc{test}(\gamma_2)$ implies $E_\textsc{dec}(\gamma_1) \leq E_\textsc{dec}(\gamma_2)$, where $E_\textsc{test}(\gamma)$ and $E_\textsc{dec}(\gamma)$ are respectively the energy consumption of the test circuit and of the complete decoder when using parameter~$\gamma$.
	\item The average message error rate in the test circuit and in the complete decoder is the same for all decoding iterations.
	\item The latency of the complete decoder is proportional to the latency of the test circuit, that is, if $T_\mathrm{dec}(\gamma)$ is the latency measured using the test circuit with parameter $\gamma$, the latency of the complete decoder is given by $\beta T_\mathrm{dec}(\gamma)$, where $\beta$ does not depend on $\gamma$.
	\end{enumerate}
	Assumption~1 is reasonable because the test circuit is very similar to a computation unit used in the complete decoder. The difference between the two is that the test circuit only instantiates one full VNP, the remaining $(d_c-1)$ VNPs being reduced to only their ``front'' part (as seen in Fig.~\ref{fig:testcircuit}), whereas the complete decoder has $d_c$ full VNPs for every CNP.
	Assumption~2 is the standard DE assumption, which is reasonable for sufficiently long codes.
	Finally, 
	it is possible for the clock period to be slower in the complete decoder, because the increased area could result in longer interconnections between circuit blocks. Even if this is the case, the interconnect length only depends on the area of the complete decoder, which is not affected by the parameters we are optimizing, and hence $\beta$ does not depend on $\gamma$.
	
	Clearly, if Assumption~1 holds and the performance of the test circuit is the same as the performance of the complete decoder, then the solution of the energy minimization is also the same. 
	The performance is composed of the three components $(p_e^{(0)}, \pres, T_\mathrm{dec})$. 
	The channel error rate $p_e^{(0)}$ does not depend on the decoder and is clearly the same in both cases.
	Because of Assumption~2, the complete decoder can achieve the same residual error rate as the test circuit when $p_e^{(0)}$ is the same.
	The latencies measured on the test circuit and on the complete decoder are not necessarily the same, but if Assumption~3 holds, and if we assume that the constant $\beta$ is known, then we can find the solution to the energy minimization of the complete decoder subject to constraints $(p_e^{(0)},\pres,T_\mathrm{dec})$ by instead minimizing the energy of the test circuit with constraints $(p_e^{(0)}, \pres, T_\mathrm{dec}/\beta)$.

	We also consider another interesting optimization problem.
	It is well known that for a fixed degree of parallelism, energy consumption is proportional to processing speed (represented here by $T_\mathrm{dec}$), which is observed both in the physical energy limit stemming from Heisenberg's uncertainty principle \cite{lloyd:2000}, as well as in practical CMOS circuits \cite{gonzalez:1996}.
	In situations where both throughput normalized to area and low energy consumption are desired, optimizing the product of energy and latency or \emph{energy-delay product} (EDP) for a fixed circuit area can be a better objective. In that case the performance constraint is stated in terms of $\bvec{P}=(p_e^{(0)}, \pres)$, and the optimization problem becomes the following: given a performance constraint $\bvec{P}=(a,b)$, minimize $E(\gammaseq) \cdot T_\mathrm{dec}(\gammaseq)$ subject to $p_e^{(0)} = a$, $\pres \leq b$, and a fixed circuit area.

\begin{table*}[t]
\centering
\caption{Energy and EDP optimization results.}
\begin{tabular}{ccllcrrrrr}
\toprule
                &     &     &     &  &     & \multicolumn{2}{c}{Standard} & \multicolumn{2}{c}{Quasi-synchronous} \\
\cmidrule{7-10}
Code & Nom. & \multicolumn{1}{c}{Norm.} & \multicolumn{1}{c}{$p_e^{(0)}$} & $\pres$ & \multicolumn{1}{c}{Latency} 
	& \multicolumn{1}{c}{Energy} & \multicolumn{1}{c}{EDP} & \multicolumn{1}{c}{Best energy} & \multicolumn{1}{c}{Best EDP}\\
family & $T_\mathrm{clk}$ & \multicolumn{1}{c}{area $\dagger$} & & & \multicolumn{1}{c}{[$\ns$]}
  	& \multicolumn{1}{c}{[$\pJ$]} & \multicolumn{1}{c}{[$\nJns$]} & \multicolumn{1}{c}{[$\pJ$]} & \multicolumn{1}{c}{[$\nJns$]} \\ 
\midrule
(3,6) & $2.0\ns$ & $1.066$ & $0.12^\ddag$ & $\leq 10^{-8}$ & $66$ & $250$ & $16.5$ & $192$ (-23\%) & $12.7$ (-23\%) \\ 
        &                &           & $0.09$ & $\leq 10^{-8}$ & $22$ & $68.2$ & $1.50$ & $45.0$ (-34\%)  & $0.98$ (-35\%)\\ 
(4,8) & $2.0\ns$ & $1.44$ & $0.09$ & $\leq 10^{-8}$ & $18$ & $98.5$ & $1.77$ & $74.9$ (-24\%) & $1.33$ (-25\%)  \\

\addlinespace
(3,30) & $3.0\ns$ & $1.099$ & $0.019^\ddag$ & $\leq 10^{-8}$ & $84.0$ & $\mathbf{883}$ & $74.2$ & $\mathbf{605}$ (-31\%) & $48.6$ (-35\%) \\
    & $2.5\ns$ & $1.135$ & $0.019^\ddag$ & $\leq 10^{-8}$ & $70.0$ & $916$ & $\mathbf{64.1}$ & $664$ (-28\%) & $\mathbf{46.5}$ (-27\%) \\
    & $3.0\ns$ & $1.099$  & $0.015$ & $\leq 10^{-8}$ & $39.0$ & $\mathbf{306}$ & $11.9$ & $\mathbf{196}$ (-36\%) & $7.35$ (-38\%) \\
    & $2.5\ns$ & $1.135$ & $0.015$ & $\leq 10^{-8}$ & $32.5$ & $324$ & $\mathbf{10.5}$ & $214$ (-34\%) & $\mathbf{6.92}$ (-34\%) \\
    
(4,40) & $3.0\ns$ & $1.522$ & $0.015$ & $\leq 10^{-8}$ & $27.0$ & $364$ & $9.83$ & $224$ (-38\%) & $5.93$ (-40\%) \\
\bottomrule

\multicolumn{10}{l}{$\dagger$ Cell area divided by the minimal area of the smallest decoder having the same code rate. $\ddag$ Approx. threshold.} \\
\end{tabular} 
\label{tbl:results}
\end{table*}
		
	\subsection{Dynamic Programming}\label{sec:dynprog}
	
	
	To solve the iteration-by-iteration energy and EDP minimization problems stated above, we adapt the ``Gear-Shift'' dynamic programming approach proposed in \cite{ardakani:2006}. The original method relies on the fact that the message distribution has a \gls{1D} characterization, which is chosen to be the error probability. By quantizing the error probability space, a trellis graph can be constructed in which each node is associated with a pair $(\tilde{p}_e^{(t)}, t)$.
	Quantized quantities are marked with tildes.
	A particular choice of $\gammaseq$ corresponds to a path $P$ through the graph, and the optimization is transformed into finding the least expensive path that starts from the initial state $(\tilde{p}_e^{(0)}, 0)$ and reaches any state $(\tilde{p}_e^{(t)},t)$ such that $\tilde{p}_e^{(t)} \leq \pres$ and the latency constraint is satisfied, if there is one.
	Note that to ensure that the solutions remain achievable in the original continuous space, the message error rates $p_e^{(t)}$ are quantized by rounding up. To maintain a good resolution at low error rates, we use a logarithmic quantization,
with 1000 points per decade.
	
	In the case of a faulty decoder, we want to evaluate the decoder's progress by tracking a complete message distribution using DE, rather than simply tracking the message error probability.
	In this case, the Gear-Shift method can be used as an approximate solver by projecting the message distribution 
	$\bvec{\pi}^{(t)}=\pmf(\mu_{i,j}^{(t)} | x_i=1)$ 
	onto the error probability space. We refer to this method as DE-Gear-Shift.
	Any path through the graph is evaluated by performing DE on the entire path using exact distributions, but different paths are compared in the projection space. As a result, the solutions that are found are not guaranteed to be optimal, but they are guaranteed to accurately represent the progress of the decoder.
	
	In the DE-Gear-Shift method, a path $P$ is a sequence of states $\{\bvec{\pi}^{(t)}\}$. As in the original Gear-Shift method, any sequence of decoder parameters $\gammaseq$ corresponds to a path. We denote the projection of a state onto the error probability space as $p_e^{(t)}= \Theta(\bvec{\pi}^{(t)})$.
	To each path $P$, we associate an energy cost $E_P$ and a latency cost $T_P$.
	A path ending at a state $\bvec{\pi}^{(t)}$ can be extended with one additional decoding iteration using parameter $\gamma$ by evaluating one DE iteration to obtain $\bvec{\pi}^{(t+1)} = f_\gamma(\bvec{\pi}^{(t)}, \bvec{\pi}^{(0)})$.
	Performing this additional iteration adds an energy cost $c_\gamma(\tilde{p}_e^{(t)}, p_e^{(0)})$ and a latency cost $T_\gamma$ to the path's cost.
	When optimizing EDP, we define the overall cost of a path $C_P$ as $C_P = E_P \cdot T_P$. When optimizing energy under a latency constraint, we define the path cost as a two-dimensional vector $C_P = (E_P, T_P)$.
	
	We use the following rules to discard paths that are suboptimal in the error probability space. Rule~1: Paths for which the message error rate is not monotonically decreasing are discarded. Rule~2: A path $P$ with cost $C_P$ is said to \emph{dominate} another path $P'$ with cost $C_{P'}$ if all the following conditions hold: 1) an ordering exists between $C_P$ and $C_P'$, 2) $C_P \leq C_P'$, 3) $\Theta(\bvec{\pi}_P) \leq \Theta(\bvec{\pi}_{P'})$, where $\bvec{\pi}_P$ denotes the last state reached by path $P$. The search for the least expensive path is performed breadth-first. After each traversal of the graph, any path that is dominated by another is discarded.
	
	When the path cost is one-dimensional, the optimization requires evaluating $O(|\Gamma| N_s)$ DE iterations, where $|\Gamma|$ is the number of operating points being considered and $N_s$ the number of quantization levels used for $\tilde{p}_e^{(t)}$. This can be seen from the fact that with a 1-D cost, Rule 2 implies that at most one path can reach a given state $\tilde{p}_e^{(t)}$. Therefore, $O(|\Gamma| N_s)$ DE iterations are required for each decoding iteration. In addition, upper bounds can be derived for the number of decoding iterations spanned by the trellis graph in terms of the smallest latency and energy cost of the parameters in $\Gamma$, and therefore it is a constant that does not depend on $|\Gamma|$ or $N_s$.
	On the other hand, when the cost is two-dimensional, the number of DE iterations could grow exponentially in terms of the number of decoding iterations. However, even in the case of a 2-D cost, an ordering exists between the costs of paths $P$ and $P'$ if $(E_P\geq E_{P'} \wedge T_P\geq T_{P'}) \vee (E_P\leq E_{P'} \wedge T_P\leq T_{P'})$, and in that case Rule~2 can be applied. In practice, for the cases presented in this paper, the discarding rules allowed to keep the number of paths down to a manageable level, even when using a 2-D cost.
	Note that an alternative to the use of a 2-D cost is to define a 1-D cost as $C_P = E_P + \kappa T_P$, and to perform a binary search for the value of $\kappa$ that yields an optimal solution with the desired latency.
	
	The algorithm can also be modified to search for parameter sequences that have other desirable properties beyond minimal energy or EDP. For example, if the decoder is implemented as a pipelined sequence of decoders, it can be desirable to favor solutions that do not require the decoder to switch its parameters too often. We can find good approximate solutions by adding a penalty to $E_P$ when the algorithm used in the current and next steps is different.
	
	\subsection{Results}\label{sec:results}

	We use DE-Gear-Shift to find good parameter sequences $\gammaseq$ for several regular ensembles with rates $\frac{1}{2}$ and $\frac{9}{10}$. 
	The parameter space $\Gamma$ consists of $(\vdd,T_\mathrm{clk})$ points with $\vdd$ from $0.70\volt$ to $1.0\volt$ in steps of $0.05\volt$ and several $T_\mathrm{clk}$ values depending on $\vdd$, in steps of $0.1\ns$.
	The standard and quasi-synchronous decoders use the same circuits.
	Parameter $\alpha$ in \eqref{eq:channelbelief} is set to $\alpha=4$ for the $(3,6)$, $(3,30)$, and $(4,40)$ decoders, and to $\alpha=2$ for the $(4,8)$ decoder. The offset parameter $C$ in Alg.~\ref{alg:loms} is set to $C=2$ for the $(4,40)$ decoder and to $C=1$ for all other decoders.
	As part of our best effort to design a good standard circuit, in the case of the $(3,30)$ decoder we present results for two circuits synthesized with different nominal $T_\mathrm{clk}$ values. The standard circuit has a lower energy consumption when synthesized with $T_\mathrm{clk}=3\ns$, while it has a lower EDP when synthesized with $T_\mathrm{clk}=2.5\ns$.

	We first run the DE-Gear-Shift solver without any path penalties to obtain the best possible parameter sequences, for both the energy and the EDP objectives. We also noticed that in some cases, adding a small algorithm change penalty allows to discover slightly better sequences.
	Note that when the objective is EDP, there is no constraint on latency. These results are summarized in Table~\ref{tbl:results}, where the energy is normalized per check node.
	Overall, we see that significant gains are possible while achieving the same channel noise, latency, and residual error requirements.
	The synthesis results show that increasing $d_v$ while keeping the rate constant leads to a significant increase in circuit area.
	Despite this, increasing the node degrees can result in a reduction of the EDP.
	For the rate $\frac{9}{10}$ ensembles, going from $d_v=3$ to $d_v=4$ decreases EDP by 6.4\% for a standard system, and by 14\% for a quasi-synchronous system. 
	However this is not the case for the rate $\frac{1}{2}$ ensembles, where $d_v=3$ has the smaller EDP.
	As expected, we can also see that much more energy is required when the channel quality is close to the ensemble's threshold.
	
	By applying a cost penalty to parameter switches, it is possible to find parameter sequences with few switches, without a large increase in cost.
	For example, for a $(3,6)$ decoder starting at $p_e^{(0)}=0.09$, a single operating condition can provide a 32\% EDP improvement, using $\gammaseq= [[0.8\volt,2.1\ns]^{11}]$. The probability of a non-zero deviation in that schedule ranges from 0.6\% to 7.2\%.
	In the case of a $(3,30)$ decoder synthesized at a nominal $T_\mathrm{clk}=2.5\ns$, for $p_e^{(0)}=0.015$ 
	the sequence $\gammaseq= [ [0.8\volt, 2.5\ns]^{12}, \allowbreak [1.0\volt, 2.5\ns] ]$ provides a 30\% EDP improvement, with non-zero deviation probabilities from 0 to 0.8\%.
	For a $(4,40)$ decoder, the single-parameter sequence $\gammaseq= [[ 0.8\volt, 2.8\ns]^{9}]$ provides a 39\% EDP improvement, with non-zero deviation probabilities from 1.6 to 4.6\%.

\section{Conclusion}\label{sec:conclusion}
We presented a method for the design of synchronous circuit implementations of signal processing algorithms that permits timing violations without the need for hardware compensation. 
We introduced a model for the deviations occurring in LDPC decoder circuits affected by timing faults that 
represent the circuit behavior accurately \cite{leduc-primeau:2016a}, 
while being independent of the iterative progress of the decoder.
In addition, we showed that in order to use density evolution to predict the performance of the faulty decoder, it is sufficient for the deviation model to have a weak symmetry property, which is more general than previously proposed sufficient properties.

We then presented an approximate optimization method called DE-Gear-Shift to find sequences of circuit operating parameters that minimize the energy or the energy-delay product. The method is similar to the previously proposed Gear-Shift method, but relies on density evolution rather than ExIT charts to evaluate the average iterative progress of the decoder.
Our results show that the best energy or EDP reduction is achieved by operating the circuit with a large number of timing violations (often with an average probability of non-zero deviation above 1\%). Furthermore, important savings can be achieved with few parameter switches, and without any compromise on circuit area or decoding performance.

In this work, we only considered delay variations associated with the signal transitions at the input of the circuit. While the energy savings that result from tolerating these variations are already significant, we ultimately see quasi-synchronous systems as an approach for tolerating the large process variations found in near-threshold CMOS circuits and other emerging computing technologies, potentially enabling energy savings of an order of magnitude.
Furthermore, we believe this approach can be extended to other self-correcting algorithms, such as deep neural networks.

\appendices
	\section{CAD Workflow}\label{sec:appendix:workflow} 
	The deviations and the energy consumption are measured directly on optimized circuit models generated by a commercial synthesis tool (\textvtt{Cadence Encounter} \cite{cadence_encounter}). We use TSMC's 65~nm process with the \emph{tcbn65gplus} cell library \cite{tcbn65gplus}.
	In order to provide a fair assessment of the improvements provided by the quasi-synchronous circuit, we first synthesize a \emph{benchmark} circuit that represents a best effort at optimizing the metric of interest, for example energy consumption. Since we do not have a specific throughput constraint for the design, we synthesize the benchmark circuit at the standard supply voltage of the library ($\vdd=1.0$V), while the clock period is chosen as small as possible without causing a degradation of the target metric.
	Second, we synthesize a \emph{nominal} circuit that will serve as the basis for the quasi-synchronous design. In this work, we use a standard synthesis algorithm for the nominal circuit, and in all the cases that we report on, the nominal and the benchmark circuits are actually the same.
	Using a standard synthesis method for the nominal circuit allows using off-the-shelf tools, but is not ideal since the objective of a standard synthesis algorithm (to make all paths only as fast as the clock period) differs from the objective pursued when some timing violations are permitted.
	For example, results in \cite{kahng:2010} show that the power consumption of a circuit can be reduced by up to 32\% when the gate-sizing optimization takes into account the acceptable rate of timing violations.
	Therefore it is possible that our results could be improved by using a different synthesis algorithm.
		
	Once the circuit is synthesized, we perform a static timing analysis of the gate-level model at various supply voltages. All timing analyses (including at the nominal supply) are performed using timing libraries generated by the \textvtt{Cadence Encounter Library Characterization} tool.
	We then use this timing information in a functional simulation of the gate-level circuit to observe the dynamic effect of path delay variations and measure the deviation statistics. 
	Any source of delay variation that can be simulated can be studied, but in this paper we focus on variations due to path activation, that is the variations in delay caused by the different propagation times required by different input transitions. 
	Note that other methods could be used to obtain the propagation delays, 
	such as the method described in \cite{pirbadian:2014} based on analytical models. In addition to speeding up the characterization, such methods allow considering the effect of process variations.
	
	Power estimation is performed by collecting switching activity data in the functional simulation and using the power estimation engine in \textvtt{Cadence Encounter}.
	However, because the circuit is operated in a quasi-synchronous manner, the clock period used to run the circuit is not necessarily the same as the nominal clock period. When that is the case, the power estimation generated by the synthesis tool cannot be used directly. First, the switching activity recorded during the functional simulation must be scaled so that it corresponds to the nominal clock period. 
	The tool's power estimation then reports the dynamic power $P_\mathrm{dyn}$ and the static power $P_\mathrm{stat}$.
	The dynamic energy consumed during one clock cycle does not depend on the clock period, whereas the static energy does. Therefore, the total energy consumed during one cycle by the quasi-synchronous circuit is given by $E_\mathrm{cycle} = P_\mathrm{dyn} T_\mathrm{clk,nom} + P_\mathrm{stat} T_\mathrm{clk}$, where $T_\mathrm{clk,nom}$ is the nominal clock period and $T_\mathrm{clk}$ is the actual clock period used to run the circuit.

	\section{Test Circuit Monte-Carlo Simulation}\label{sec:appendix:testcircuit}
	
	A suitable test circuit for a row-layered decoder architecture consists in implementing a single check node processor, as well as the necessary logic taken from the variable node processor block to send $d_v$ messages to the CNP, and receive one message from the CNP. This test circuit is shown in Fig.~\ref{fig:testcircuit}. It re-uses logic blocks that are found in the complete decoder, ensuring the accuracy of the deviation and energy measurements, and minimizing design time.
	
	The test circuit is used to evaluate the decoder's computation tree (shown in Fig.~\ref{fig:LDPC_tree_dev_model}).
	The VNP with index $1$, shown at the top, is always mapped to the VN that is at the head of the computation tree, while the VNPs at the bottom of the figure are mapped to different VNs as the CNP is successively mapped to each CN neighbor of the head VN.
	
	At any given clock cycle, a \emph{VNP front} block is mapped to a particular VN $i$. For illustrative purposes, we simply index the VN neighbors from $1$ to $d_c$, even if the VNs mapped to the bottom VNPs actually change at each layer.
	Each \emph{VNP front} block takes as input the previous belief total of that VN $\Lambda'_i$, and the previous CN-to-VN message corresponding to layer $\ell$, $\lambda_{i,J(i,\ell)}^{(t-1)}$.
	
	To perform the Monte-Carlo simulation, a \emph{VNP front} circuit block with index $i$ must send a message $\mu_i^{(t)}$, randomly generated according to a \gls{1D} normal distribution with error probability $p_e^{(t)}$. However, the only inputs that are controllable are $\Lambda'_i$ and $\lambda_{i,J(i,\ell)}^{(t-1)}$.
	To simplify the Monte-Carlo simulation, we disregard the true distribution of $\lambda_{i,J(i,\ell)}^{(t-1)}$ and generate it according to a \gls{1D} normal distribution.
	We also introduce another simplification: we assume that messages received at a VN only modify the total belief at the end of the iteration, as would be the case when using a flooding schedule. As a result, the messages $\mu_i^{(t)}$ are identically distributed with error rate parameter $p_e^{(t)}$ for all $\ell$. Note that these simplifications are not necessary, and they could be removed at the cost of a slightly more cumbersome Monte-Carlo simulation.
	
	To generate inputs with the appropriate distribution, we use the fact that 
	$\Lambda'_i = \mu^{(t)}_{i,J(i,\ell)} + \lambda^{(t-1)}_{i,J(i,\ell)}$. 
	On a cycle-free Tanner graph, $\mu^{(t)}_{i,J(i,\ell)}$ and $ \lambda^{(t-1)}_{i,J(i,\ell)}$ are independent, 
	but naturally $\Lambda'_i$ and $\lambda^{(t-1)}_{i,J(i,\ell)}$ are not. Therefore, we generate $\mu^{(t)}_{i,J(i,\ell)}$ and $\lambda^{(t-1)}_{i,J(i,\ell)}$ and sum them to obtain $\Lambda'_i$. 

	To complete the DE iteration, we want to measure an extrinsic message belonging to the next iteration. Because we assume a flooding schedule, this extrinsic message can be obtained by summing any set of $(d_v-1)$ messages in the current iteration.
	To achieve this, we start a DE iteration by setting $\Lambda'_1 \gets 0$ and $\lambda^{(t-1)}_{1,J(1,\ell)} \gets 0$ for all $\ell$. The desired extrinsic message then corresponds to the total belief output of the circuit $\Lambda^{(t)}_1$ after $d_v-1$ layers have been evaluated.

	Just like the processor used in the complete decoder, the test circuit has one input and one output register, as well as one internal pipeline register, for a latency of 3 clock cycles. In order to keep the pipeline fed, several distinct computation trees are evaluated in parallel during the Monte-Carlo simulation.
	
	\begin{figure}[tbp]
	\begin{center}
	\includegraphics[width=2in]{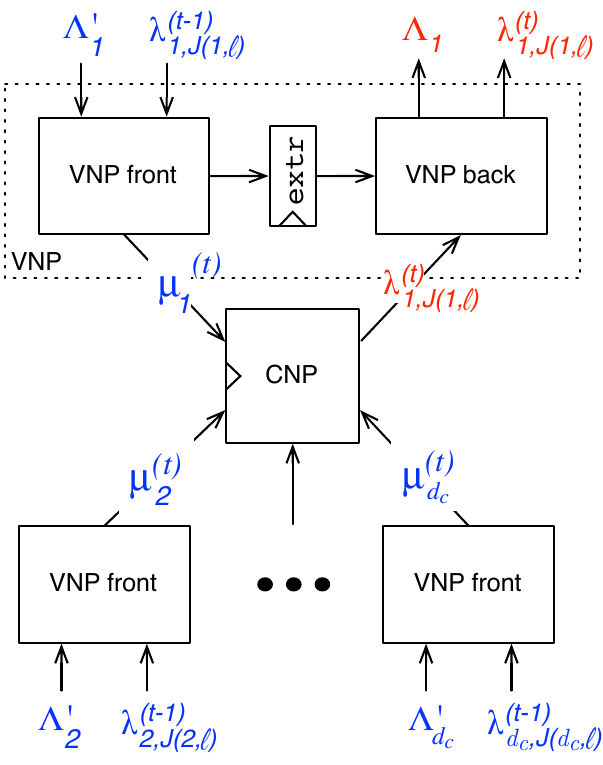}
	\caption{Block diagram of the test circuit.}
	\label{fig:testcircuit}
	\end{center}
	\end{figure}

\section*{Acknowledgements}
The authors wish to thank CMC Microsystems for providing access to the Cadence tools and TSMC 65nm CMOS technology, and Gilles Rust for advice on Cadence tools and cell library characterization.

\bibliography{IEEEabrv,computing_refs.bib,article_refs.bib}
\bibliographystyle{IEEEtran}

\end{document}